\documentclass[journal]{IEEEtran}
\usepackage{amsmath,amssymb,amsfonts}
\usepackage{amsthm}
\usepackage{calligra}
\usepackage{algorithmicx}
\usepackage[linesnumbered,ruled,vlined]{algorithm2e}
\usepackage{bm}
\usepackage{cite}
\usepackage{hyperref}
\usepackage{enumerate}
\usepackage[table]{xcolor}
\usepackage{graphicx}
\usepackage{epstopdf}
\usepackage{cite,url,epsfig}
\usepackage{epstopdf}
\usepackage{caption}
\usepackage{comment}
\usepackage{booktabs}
\usepackage{subfigure}

\newtheorem{definition}{Definition}

\newtheorem{lemma}{Lemma}

\newtheorem{remark}{Remark}

\begin{document}

\title{Diffusion Model-based Incentive Mechanism with Prospect Theory for Edge AIGC Services in 6G IoT}

\author{Jinbo Wen, Jiangtian Nie, Yue Zhong, Changyan Yi, Xiaohuan Li, Jiangming Jin,\\ Yang Zhang*, and Dusit Niyato, \textit{Fellow, IEEE}

\thanks{ 
 
J. Wen, C. Yi, and Y. Zhang are with the College of Computer Science and Technology, Nanjing University of Aeronautics and Astronautics, China (e-mails: jinbo1608@nuaa.edu.cn; changyan.yi@nuaa.edu.cn; yangzhang@nuaa.edu.cn).

J. Nie and D. Niyato are with the School of Computer Science and Engineering, Nanyang Technological University, Singapore (e-mails: jnie001@e.ntu.edu.sg; dniyato@ntu.edu.sg).

Y. Zhong is with the School of Automation, Guangdong University of Technology, China (e-mail: 3220001516@mail2.gdut.edu.cn).

X. Li is with the School of Information and Communication, Guilin University of Electronic Technology, China (e-mail: lxhguet@guet.edu.cn).

J. Jin is with Qingcheng AI, China (e-mail: jiangming.jin@outlook.com).

\textit{*Corresponding author: Yang Zhang}

	} 
}

\maketitle

\begin{abstract}
The fusion of the Internet of Things (IoT) with Sixth-Generation (6G) technology has significant potential to revolutionize the IoT landscape. With the ultra-reliable and low-latency communication capabilities of 6G, 6G-IoT networks can transmit high-quality and diverse data to enhance edge learning. Artificial Intelligence-Generated Content (AIGC) harnesses advanced AI algorithms to automatically generate various types of content. The emergence of edge AIGC integrates with edge networks, facilitating real-time provision of customized AIGC services by deploying AIGC models on edge devices. However, the current practice of edge devices as AIGC Service Providers (ASPs) lacks incentives, hindering the sustainable provision of high-quality edge AIGC services amidst information asymmetry. In this paper, we develop a user-centric incentive mechanism framework for edge AIGC services in 6G-IoT networks. Specifically, we first propose a contract theory model for incentivizing ASPs to provide AIGC services to clients. Recognizing the irrationality of clients towards personalized AIGC services, we utilize Prospect Theory (PT) to capture their subjective utility better. Furthermore, we adopt the diffusion-based soft actor-critic algorithm to generate the optimal contract design under PT, outperforming traditional deep reinforcement learning algorithms. Our numerical results demonstrate the effectiveness of the proposed scheme.

\end{abstract} 

\begin{IEEEkeywords}
6G-IoT networks, edge AIGC, contract theory, prospect theory, generative diffusion models.
\end{IEEEkeywords}
\IEEEpeerreviewmaketitle

\section{Introduction}\label{Intro}
Recently, the Internet of Things (IoT) has faced challenges such as diverse application scenarios, increasing data volumes, and growing computational demands, posing obstacles for current wireless communication technologies to effectively support large-scale IoT deployments \cite{9369324}. Thanks to further improvements in coding techniques and radio interface modulation, Sixth-Generation (6G) networks will be significantly faster than previous generations, allowing users to connect to each other everywhere\cite{dang2020should}. With the imminent arrival of 6G technology, which heralds a transformative leap in wireless communication networks, the integration of 6G with IoT holds the potential to address these challenges \cite{vaezi2022cellular}. Since 6G has the advantage of ultra-reliable, large-scale coverage, and low-latency communications\cite{wen2022optimal}, 6G-IoT networks can provide ultra-high speed services for users, and enhance the connectivity and scalability of the IoT ecosystem\cite{you2021towards}. Moreover, 6G-IoT networks can offer robust support for edge computing, increasing efficiency by locating computing and data storage close to the data source \cite{qadir2023towards}. This approach enables the real-time provisioning of diverse data types for training Artificial Intelligence (AI) models\cite{qadir2023towards}, thereby facilitating intelligent decision-making across various IoT application scenarios. For instance, 
%in smart transportation systems, 6G-IoT networks can collect real-time traffic data, utilizing AI algorithms to analyze the data and optimize traffic flow\cite{wang20206g}. 
in smart healthcare, 6G-IoT networks play a pivotal role in supporting applications such as remote patient monitoring and personalized medicine by furnishing a continuous stream of data for AI-driven analysis and decision-making \cite{9922093}.

As a novel paradigm for the production and manipulation of data, AI-Generated Content (AIGC) has attracted extensive attention from academia and industry\cite{xu2024unleashing}. Unlike discriminative AI that focuses on analyzing or classifying existing data, Generative AI (GAI), as a core technique of AIGC, can learn patterns and structures from existing data and autonomously generate novel content, such as text, images, audio, and even synthetic data\cite{xu2024unleashing}. GAI constitutes a range of different models, each with its own unique advantages and practical applications\cite{du2023deep}. For example, Generative Diffusion Models (GDMs) excel in image generation \cite{du2023enabling} and network optimization\cite{du2023deep}. The applications of AIGC span a wide range of domains, showcasing its versatility in creative content generation, language translation, and remarkable image synthesis\cite{cao2023comprehensive}. ChatGPT, as a transformer-based large language model, excels in natural language dialogues, while DALL-E generates original and realistic images based on user prompts. Thanks to these prominent capabilities, AIGC has the capacity to drive the progression of 6G-IoT networks\cite{wen2024generative}.
% GDM exemplifies a sophisticated approach in the realm of GAI, showcasing exceptional efficacy and versatility in various tasks and domains. 
%In the future, the integration of GAI and 6G technology presents new prospects, which enables GAI to assimilate personalized data knowledge from the extensive network of connected 6G IoT devices \cite{chen2024gainnet}. Simultaneously, the formidable generation capabilities of GAI can offer diverse AI Generated Content (AIGC) services to the 6G IoT devices. 

%第三段（Challenge）：然后谈谈目前GAI主要还是中心化架构，存在时延问题，所以可以将GAI与边缘计算结合，引出Edge AIGC的概念及特点（关于Edge AIGC的概念可以参考GPT,给一个大致的定义）。接着举一个Edge AIGC提供服务的例子（可以去看寅秋师兄那篇magainze第6页和杜博的文章），引出目前这些ASPs是work for free，同时，用户对AIGC服务存在个性化需求与主观感受，所以需要制定一个User-centric incentive mechanism
Although AIGC is believed to have the potential to revolutionize existing production processes\cite{xu2024unleashing}, there is a significant AIGC service latency problem due to the existing centralized AIGC framework, and users currently accessing AIGC services on mobile devices lack the support of computing and storage resources\cite{xu2024unleashing}. To address these challenges, integrating GAI with mobile edge computing gives rise to the concept of edge AIGC \cite{xu2023edge, du2023enabling}. Edge AIGC refers to the deployment of GAI models on edge devices, and edge devices acting as AIGC Service Providers (ASPs) can provide personalized and low-latency AIGC services for users\cite{peng2020multi}, effectively enriching user experiences \cite{du2023enabling}. This paradigm shift to edge deployment enables efficient content generation, reduces service latency, and optimizes network resources\cite{wen2023freshness}. Currently, some crowdsourcing platforms for edge AIGC services have been proposed to conduct AIGC inferences for clients\cite{liu2023deep}. These platforms accommodate hundreds of ASPs contributing their computational and storage resources for AIGC service provisions. However, these ASPs have no financial compensation, potentially impacting their long-term commitment to the platform \cite{liu2023deep}. Moreover, clients might not be aware of ASPs' private information (e.g., local AIGC model complexity and quality) due to information asymmetry and possess subjective perceptions concerning personalized AIGC services, which affects the service quality of ASPs to meet the demands of clients \cite{liu2023deep, fan2023learning}. Therefore, it is necessary to design a user-centric incentive mechanism to motivate ASPs to sustainably provide high-quality AIGC services. Some studies have been conducted to design incentive mechanisms to motivate ASPs\cite{lai2023resource, liu2023deep}. For instance, the authors in \cite{liu2023deep} proposed a diffusion-based contract model to incentivize ASPs for mobile AIGC services. However, none of the studies considers the subjective behavior of decision-makers, which may make the incentive mechanism unrealistic.

%第四段，就是总结工作，可以去参考我那篇TCCN
To address the above challenges, in this paper, we develop a user-centric incentive mechanism framework for edge AIGC services in 6G-IoT networks. Specifically, we first propose a contract theory model to incentivize ASPs to provide AIGC services to clients. Considering that the client may behave irrationally when facing uncertain and risky circumstances, e.g., uncertainty regarding the reliability and consistency of AIGC services, we utilize Prospect Theory (PT) to capture the subjective utility of the client, enhancing the reliability of the proposed contract model in practice\cite{kang2023blockchain}. Given the demonstrated effectiveness of the GDM-based approach in generating optimal contracts \cite{10409284,du2023deep}, we utilize GDMs to generate optimal contracts, thus motivating ASPs for edge AIGC services. \textit{To the best of our knowledge, this is the first work to adopt GDMs for the user-centric incentive mechanism design under PT}. The main contributions of this paper are summarized as follows:

\begin{itemize}
    \item We design a novel user-centric incentive mechanism framework for edge AIGC services in 6G-IoT networks. The framework reveals the symbiotic interaction between edge AIGC and 6G-IoT networks, in which 6G-IoT networks can provide large amounts of various data for GAI model fine-tuning and promote edge AIGC services, and edge AIGC can empower intelligent IoT applications based on the ability of real-time decision-making.
    \item We propose a contract theory model to motivate ASPs to provide AIGC services to clients, which can mitigate information asymmetry between ASPs and a client. Given the potential for the irrational behavior of the client in response to personalized AIGC services in uncertain and risky environments, we utilize PT to more accurately capture the utility of the client, and the optimal contract is determined by maximizing its subjective utility.
    \item We adopt a diffusion-based Soft Actor Critic (SAC) algorithm to capture the high-dimensional and complex of the formulated problem. This algorithm adapts to the continuous action space and can generate optimal contracts under PT. To evaluate the effectiveness of the proposed contract model and the algorithm, we perform numerical analysis and demonstrate that the proposed GDM-based scheme outperforms Deep Reinforcement Learning (DRL)-based schemes.

\end{itemize}

The rest of this paper is organized as follows. In Section \ref{RW}, we review related literature and outline the preliminaries of PT and GDMs. In Section \ref{framework}, we introduce the user-centric incentive mechanism framework for edge AIGC services in 6G-IoT networks. In Section \ref{Problem}, we propose a contract model under PT to motivate ASPs to provide AIGC services. In Section \ref{Optimial_Contract}, we propose the diffusion-based SAC algorithm for optimal contract design under PT. In Section \ref{Results}, we conduct an evaluation of the proposed model and the algorithm. Section \ref{Conclusion} concludes this paper.

\section{Related Work and Preliminaries}\label{RW}
In this section, we delve into state-of-the-art technological domains, investigating the convergence of 6G and IoT, the emergence of edge AIGC, the intricacies of diffusion-based incentive mechanisms, and the foundational principles underlying PT and GDMs.

\begin{table}[t]
\renewcommand{\arraystretch}{1.6}
\caption{Key Mathematical Notations of this Paper}
% \centering
\begin{tabular}{m{0.8cm}|m{7.0cm}} %14.5
\toprule[1.5pt]
\hline
\multicolumn{1}{c|}{\textbf{Notation}}  & \multicolumn{1}{c}{\textbf{Definition}} \\ \hline
$\theta_k$ &  The $k$-th type ASP, associating with its local AIGC model  \\ \hline
$L_k$ & Edge AIGC service latency requirement to type-$k$ ASPs \\ \hline
$R_k$ & Reward to type-$k$ ASPs for edge AIGC service provisions  \\ \hline
$a$ & Unit resource cost of edge AIGC service provisions  \\ \hline
$L_{max}$ & Maximum tolerant edge AIGC service latency   \\ \hline
$f$ & Weight parameter about the incentive $R_k$ of type-$k$ ASPs   \\ \hline
$\eta$ & Loss aversion coefficient\\ \hline
$U_{ref}$ & Reference point for all types of ASPs\\ \hline
$e_1$ & Parameter regarding the quality of AIGC model inference  \\ \hline
$e_2$ & Parameter regarding the latency spent for edge AIGC service delivery \\ \hline
$z_1$ & Factor indicating the effect of inference quality\\ \hline
$z_2$ & Factor indicating the effect of service latency \\ \hline
$Q_k$ & The proportion of type-$k$ ASPs in the AIGC market \\ \hline
$\pi_\omega$ & Contract design policy with parameters $\omega$\\ \hline
$\boldsymbol{\epsilon}_\omega$ & Contract generation network with parameters $\omega$\\ \hline
$q_\varphi$ & Contract quality network with parameters $\varphi$\\\hline
$\pi_{\omega '}$ & Target contract design policy with parameters $\omega '$\\\hline
$\boldsymbol\epsilon '_{\omega '}$ & Target contract generation network of $\boldsymbol{\epsilon}_\omega$ with parameters $\omega '$\\\hline
$q_{\varphi '}'$ & Target contract quality network of $q_\varphi$ with parameters $\varphi '$\\\hline
$\Phi^{0}$ & Optimal contract design \\\hline
% $\bar{\boldsymbol{{h}}}^{(graph)}$& The graph embedding in GAT model\\\hline
\bottomrule[1.5pt]
\end{tabular}
\end{table}

\subsection{6G and Internet of Things}
% 简单介绍一下6G与IoT融合的好处，可查阅一下文献 or GPT
% 列举三到四个6G与IoT融合的工作
% The Internet of Things (IoT) refers to a network of physical objects  embedded with sensors, software, and connectivity capabilities that enable them to collect and exchange data over the internet. These objects can range from everyday devices like smartphones, wearables, and home appliances to industrial machinery, vehicles, and infrastructure components. However,
The continuous upgrading of IoT applications presents challenges such as application proliferation, data intensity, and computational requirements\cite{9369324}. Consequently, the current wireless communication technology fails to adequately meet the needs of large-scale IoT applications, and wireless communications and networking require technological advancement and evolution toward 6G networks\cite{9369324}. With the advantages of 6G, such as ultra-low latency communication, exceptionally high throughput, and satellite-based customer services, the integration of 6G with IoT has achieved a revolutionary leap in connectivity of IoT devices, providing ultra-high speed services, thus enhancing user experience and service quality in the IoT ecosystem\cite{you2021towards}. In recent years, scholars have increasingly focused on integrating 6G and IoT, exploring the potential synergies between them\cite{9044345,mao2021optimizing,9509294,khalid2023reconfigurable}. In \cite{9044345}, the authors proposed a novel AI-based approach for designing an adaptive security specification approach for 6G IoT networks, which can address the challenges related to data privacy and confidentiality in IoT by considering the connectivity of IoT devices to cellular networks through various frequency bands. The authors in \cite{mao2021optimizing} %underscored IoT challenges in remote areas attributable to satellite limitations, emphasizing the pivotal role of edge computing, machine learning, and energy harvesting for 6G. They 
highlighted the integration of Unmanned Aerial Vehicles (UAVs) and satellites to furnish edge and cloud computing services to wireless-powered IoT devices, leveraging Terahertz bands and deep learning for task optimization, thereby showcasing the synergy between 6G and IoT technologies. In \cite{9509294}, the authors scrutinized the convergence of 6G and IoT, delving into emerging opportunities and technologies for future IoT applications such as autonomous driving IoT, healthcare IoT, and industrial IoT. The above works extend the role of 6G in facilitating IoT applications and provide insights into potential directions for further exploration. 
% The authors in \cite{ahad2023perspective} proposed a vision of future intelligent healthcare that combines 6G and IoT to overcome current limitations such as network performance, security issues, and so on. 
%In \cite{mishra20236g}, the authors examined the potential of 6G-IoT-enable communication networks in facilitating sustainable smart cities, presenting an intelligent framework. They also outlined the challenges and identified research directions for future 6G-IoT wireless communication studies within the realm of smart cities.

\subsection{Edge AI-Generated Content}
%involves utilizing AI models like ChatGPT, DALL-E-2, and Codex
% 简单介绍一下AIGC，然后目前中心化架构会导致服务时延增加，所以出现Edge AIGC,然后介绍一下Edge AIGC的工作，可以重点关注我们组里面（丙坤那篇）和dusit那边的工作，大概两到三篇，然后再介绍一到两篇别人的工作
AIGC is an automated method to efficiently generate various forms of content such as images, music, and natural language, by extracting intent information from human instructions\cite{cao2023comprehensive}. Driven by large-scale models and multi-model interactions, AIGC has recently made significant progress, enhancing generative results and bringing new challenges and future avenues for exploration in the discipline \cite{10268594}. However, cloud-based AIGC pre-training models, e.g., GPT-3 for ChatGPT and GPT-4 for ChatGPT Plus, lead to high latency due to their remote natures, which drive a shift towards deploying interaction-intensive AI services on edge networks, termed as edge AIGC, to mitigate service latency \cite{xu2024unleashing}. Edge AIGC integrates GAI models with mobile edge computing\cite{xu2023edge}, enabling the provision of low-latency AI services by executing AI models on mobile devices or edge servers. The authors in \cite{10268594} conducted a comprehensive review of the developments of GAI and edge-cloud computing, and explored the technical challenges associated with scaling up solutions using edge-cloud collaborative systems through two exemplary GAI applications. 
% The authors in \cite{10172151} introduced a novel collaborative distributed diffusion-based AIGC framework to advance the ubiquity of AIGC services. The proposed framework streamlines the execution of AIGC tasks, enhancing the optimization of edge computation resource utilization. 
In \cite{10172151}, a novel collaborative distributed diffusion-based AIGC framework was introduced to promote the realization of ubiquitous AIGC services. The framework can balance the computational load, reduce latency, and achieve high-quality content generation, streamlining task execution and optimizing edge computation resource usage. In \cite{wen2023freshness}, the authors proposed a contract theory model based on age of information to incentivize fresh data sharing among UAVs, thus ensuring the quality of mobile AIGC services. Although edge AIGC achieves low-latency services, challenges such as resource allocation optimization and latency reduction remain, highlighting the necessity for further research to fully exploit the potential of AIGC in mobile edge networks.

% The field of AIGC has witnessed significant advancements, fueled by large-scale models and multimodal interactions. The integration of GAI models with mobile edge computing has facilitated the emergence of low-latency AI services. Despite significant progress, challenges such as optimizing resource allocation and reducing latency persist, underscoring the need for additional research to fully harness the potential of AIGC in mobile edge networks.

\subsection{Diffusion-based Incentive Mechanisms}
% 简单介绍一个GDM，然后简单说明一下GDM求解network optimization problem的原理，列举三到四篇关于GDM求解incentive mechanisms的工作，目前这个工作应该还不多（可以稍微提一下这个），然后列举丙坤（斯坦伯格），我（契约），杜博那边的工作，但上述工作均没有考虑主方面对不确定环境下的行为（去看看展望理论的应用场景，我那篇tccn），然后引出我们要用展望理论
GAI represents a paradigm shift beyond the confines of traditional AI. While traditional AI models primarily analyze or categorize existing data, GAI exhibits the remarkable capability to generate entirely new data, spanning various domains such as text, images, and audio \cite{du2023deep}. As a prominent example in this area, GDMs can generate optimal solutions without relying on annotated data or expert knowledge\cite{lai2023resource}, making them invaluable in various incentive mechanisms, such as Stackelberg game, auction theory, and contract theory, to derive optimal strategies \cite{du2023deep}. Currently, a few works have been conducted on utilizing GDMs for incentive mechanism design\cite{lai2023resource,wen2024generative,nguyen2024generative}. In \cite{lai2023resource}, the authors introduced a new paradigm named generative mobile edge networks by integrating GAI with mobile edge networks. They proposed a Stackelberg model to solve the resource allocation problem and leveraged GDMs to obtain optimal solutions. The authors in \cite{wen2024generative} presented a novel framework for a secure incentive mechanism based on GAI, wherein they employ GDMs for designing the incentive mechanism. Additionally, they utilized GDMs to generate efficient contracts, thus motivating users to contribute high-quality sensing data. The authors in \cite{nguyen2024generative} put forward a novel approach by leveraging GDMs to achieve optimization in the blockchain. They suggested that GDMs can be effectively used to generate high-quality data for evaluation purposes, as well as to generate solutions that optimize blockchain consensus mechanisms and network parameters. The aforementioned works have made contributions to incorporating GDMs into incentive mechanism design. However, a significant aspect that remains unexplored in these endeavors is the behavioral considerations of decision-makers in uncertain environments when participating in incentive mechanisms. In this paper, we focus on adopting GDMs for incentive mechanism design, considering the subjective behavior of decision-makers.

\subsection{Preliminaries of Prospect Theory and Generative Diffusion Models}
%介绍一下展望理论
\subsubsection{Prospect Theory}
According to traditional decision-making theory, decision-makers are considered to be consistently rational, optimizing their decision-making process to maximize personal utility based on Expected Utility Theory (EUT)\cite{ebrahimigharehbaghi2022application}. However, in uncertain and risky environments, decision-makers may deviate from rationality and exhibit irrational behavior. In such contexts, PT offers a more suitable framework for understanding decision-making as it accounts for the impact of uncertainty and risk on the choices of individuals \cite{kang2023blockchain}. 
% The authors in \cite{huang2021efficient} presented a contract-based incentive mechanism design for Parked Vehicles in a mobile edge computing scenario, where an offloading service provider acts on behalf of offloading users. Drawing upon contract theory and PT, the authors modeled the subjective evaluation of computing offloading utility by the users. They aimed to derive an optimal contract that maximizes subjective utility, considering the presence of information asymmetry. 
We introduce the effect of two key notions from PT as follows:
\begin{itemize}
\item \textit{Probability weighting effect:} PT differs from traditional decision-making theories, i.e., EUT, by incorporating subjective probabilities to determine the weight assigned to each potential outcome. Subjective probability, derived from objective probability, introduces a psychological bias that leads to underestimating high-probability events and overestimating low-probability events \cite{ye2022incentivizing}. 
\item \textit{Utility framing effect:} In PT, decision-makers utilize reference points to categorize outcome returns as either gains or losses. For instance, a decision-maker may establish a reference point based on a specific profit goal. If the actual outcome falls short of this goal, it is regarded as a loss, whereas surpassing the goal is seen as a gain \cite{kang2023blockchain,zhang2018wireless}.
\end{itemize}

The utility of EUT is defined as $U_{EUT}=\sum_{i=1}^I Q_i U_{i,EUT}$, where $Q_i$ and $U_{i,EUT}$ are the objective probability and the utility of alternative $i$, respectively, and $I$ is the total number of alternatives. Based on the above two key notions of PT, the utility of PT is expressed as $U_{PT}=\sum_{i=1}^I H(Q_i) U_{i,PT}$, where $H(\cdot)$ is the inverse S-shape probability weighting function of the objective probability $Q$. Specifically, $H(Q_i)=\exp(-(-\log(Q_i))^\alpha)$ is defined as the subjective probability of alternative $i$, where $\alpha$ represents a rational coefficient that how the subjective evaluation distorts objective probabilities \cite{eec14168-5714-3ca8-b073-d038266f2734}. Thus, $U_{i,PT}$ is expressed as
\begin{equation} \label{PT}
U_{i,PT}= \left\{ \begin{aligned}
(U_{i,EUT}-U_{i, ref})^{\beta^{+}},\: U_{i,EUT}\geq U_{i, ref},\\
- \gamma (U_{i,ref}-U_{i,EUT})^{\beta^{-}},\: U_{i,EUT}<U_{i,ref},\\
\end{aligned} \right.
\end{equation}
where $\beta^{+}, \beta^{-}\in (0,1]$ are weighting factors representing gain and loss distortion, respectively, $\gamma\geq 0$  is a loss aversion coefficient, and $U_{i, ref}$ is a reference point for classifying the utility of $U_{i,EUT}$ into either gain or loss \cite{kang2023blockchain}.

\subsubsection{Generative Diffusion Models}
GDM is extensively employed in image generation and also shows strong potential in decision-making scenarios, which can effectively represent complex dynamics. Notably, GDMs demonstrate scalability over the long term and integrate additional conditional variables, such as constraints. In the realm of DRL, GDMs serve as a policy representation method, adept at capturing multi-modal action distribution and enhancing the performance of offline RL tasks \cite{du2023deep}. Based on the initial input, GDMs gradually introduce Gaussian noise through a forward diffusion process and utilize a denoising network to iteratively approximate real sample $x \sim q(x)$ through a series of estimation steps, where $q(x)$ denotes the data distribution \cite{du2023deep,yang2023diffusion}. Subsequently, the denoising network undergoes training to reverse the noise process and restore the data and content, enabling the generation of new data. The detailed description of the forward and reverse diffusion process is as follows:
\begin{itemize}
    \item \textit{Forward diffusion process:} Considering a data distribution $x_0 \sim q(x_0)$, the forward process can be modeled as a Markov process with $T$ steps. After $T$ interactions, a systematic addition of Gaussian noise is performed on the initial sample $x_0$, leading to the generation of a series of samples $x_1,x_2,\ldots,x_T$ with transition kernel $q(x_t|x_{t-1})$.  By employing the chain rule of probability and leveraging the Markov property\cite{ho2020denoising}, the joint distribution of $x_1,x_2,\ldots,x_T$ conditioned on $x_0$ can be deposed $q(x_1,x_2,\ldots,x_T|x_0)$ into 
    \begin{equation}
        q(x_1,x_2,\ldots,x_T|x_0)=\prod_{t=1}^T q(x_t|x_{t-1}),
    \end{equation}
    where
    \begin{equation}
        q(x_t|x_{t-1})=\mathcal{N}(x_t;\boldsymbol\mu_t=\sqrt{1-\delta_t}x_{t-1},\boldsymbol\Sigma_t=\delta_t\textbf{I}),
    \end{equation}
    where $\boldsymbol\mu_t$ and $\boldsymbol{\Sigma}_t$ represent the mean and variation of the normal distribution, respectively, $\textbf{I}$ is the identity matrix, and $\delta_t\in(0,1)$ is a hyperparameter chosen before model training. By defining $\chi_t=1-\delta_t$ and $\bar{\chi_t}=\prod_{j=0}^t\delta_j$, $q(x_t|x_{t-1})=\mathcal{N}(x_t;\sqrt{\bar{\chi_t}}x_0,(1-\chi_t)\textbf{I})$ \cite{zhang2023text}. Given $x_0$, by sampling the Gaussian vector $\boldsymbol{\epsilon}_0,\ldots,\boldsymbol{\epsilon}_{t-1} \sim \mathcal N (\textbf{0}, \textbf{I})$ and applying the transformation, where $\textbf{0}$ represents that the mean of the normal distribution is zero, $x_t$ can be obtained as 
    \begin{equation}
        x_t=\sqrt{\bar{\chi_t}}x_0+\sqrt{(1-\bar\chi_t)}\boldsymbol{\epsilon}_0,
    \end{equation}
    
    \item \textit{Reverse diffusion process:} By learning the inverse distribution $q(x_{t-1}|x_t)$, $x_t$ can be sampled from the standard normal distribution $\mathcal{N}(\textbf{0},\textbf{I})$ by an inverse process to sample from $q(x_0)$. However, accurately estimating the statistical properties of $q(x_{t-1}|x_t)$ involves a complex computation of the data distribution, which is a challenging task. Therefore, the parametric model $p_\omega$ can be used to approximate the estimation of $q(x_{t-1}|x_t)$\cite{ho2020denoising}, i.e.,
    \begin{equation}\label{p}
        p_\omega(x_{t-1}|x_t)=\mathcal{N}(x_{t-1};\boldsymbol\mu_\omega(x_t,t),\boldsymbol\Sigma_\omega(x_t,t)),
    \end{equation}
    where $\omega$ represents the model parameters, and the mean $\boldsymbol\mu_\omega(x_t,t)$ and the variance $\boldsymbol\Sigma_\omega(x_t,t)$ are parameterized by deep neural networks. Thus, the trajectory from $x_T$ to $x_0$ is expressed as \cite{du2023deep}
    \begin{equation}
        p_\omega(x_0,x_1,\ldots,x_T)=p_\omega(x_T)\prod_{t=1}^Tp_\omega(x_{t-1}|x_t).
    \end{equation}
    Adding conditional information (e.g., $g$) during the denoising process, $p_\omega(x_{t-1}|x_t,g)$ can be modeled as a noise prediction model, and the covariance matrix is fixed as
    \begin{equation}\label{sigma}
        \boldsymbol\Sigma_\omega(x_t,g,t)=\delta_t\textbf{I},
    \end{equation}
    and the mean is constructed as
    \begin{equation}\label{mu}
        \boldsymbol\mu_\omega(x_t,g,t)=\frac{1}{\sqrt{\chi_t}}\bigg(x_t-\frac{\delta_t}{\sqrt{1-\bar{\chi_t}}}\boldsymbol\epsilon_\omega(x_t,g,t)\bigg).
    \end{equation}
    First sample $x^T \sim \mathcal N (\textbf{0}, \textbf{I})$, then sample from the reverse diffusion chain parameterized by $\omega$ as
    \begin{equation}\label{sample}
        x_{t-1}|x_t=\frac{x_t}{\sqrt{\chi_t}}-\frac{\delta_t}{\sqrt{\chi_t(1-\bar{\chi_t})}}\boldsymbol\epsilon_\omega(x_t,g,t)+\sqrt{\delta_t}\boldsymbol\epsilon,
    \end{equation}
    where $\boldsymbol\epsilon \sim \mathcal N (\textbf{0}, \textbf{I})$. By ignoring specific weight terms, the original loss function can be simplified as\cite{ho2020denoising}
    \begin{equation}
        {L}_t=\Bbb{E}_{x_0,t,\boldsymbol\epsilon}\big[\|\boldsymbol\epsilon-\boldsymbol\epsilon_\omega(\sqrt{\bar{\chi_t}}x_0+\sqrt{1-\bar{\chi_t}}\boldsymbol\epsilon,t)\|^2\big].
    \end{equation}
    %This indicates that the model is not predicting the mean of the distribution, but rather the noise $\boldsymbol\epsilon$ at each time step $t$.
\end{itemize}

\begin{figure*}[t]
\centering
\includegraphics[width=0.95\textwidth]{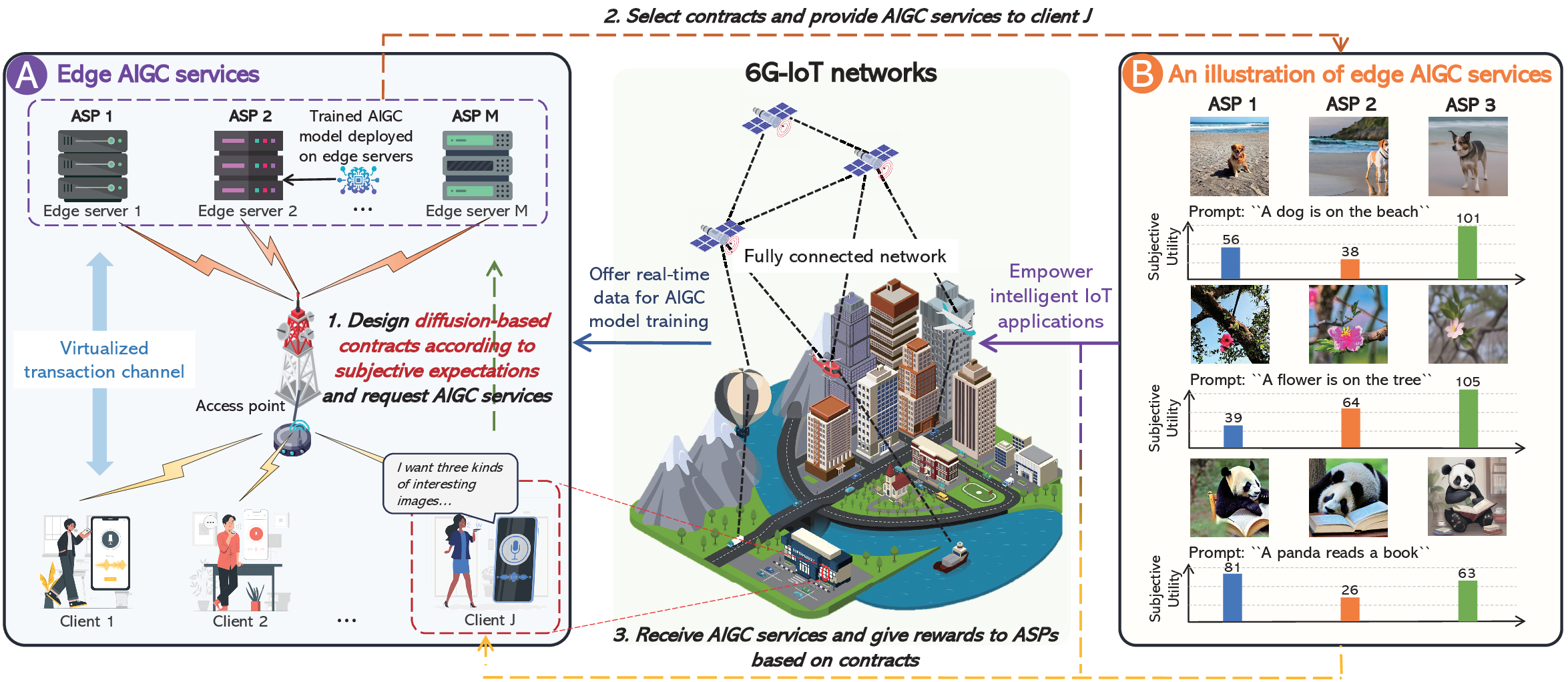}
\captionsetup{font=footnotesize}
\caption{A user-centric incentive mechanism framework for edge AIGC services in 6G-IoT networks. \textbf{Part A} is the network architecture of ASPs employing edge servers to deploy AIGC models for providing AIGC services to clients; \textbf{Part B} shows an illustration of edge AIGC services and presents variation in the subjective utility of a client from the same text prompt on various ASPs.}
\label{System}
\end{figure*}

\section{User-Centric Incentive Mechanism Framework for Edge AIGC Services in 6G-IoT Networks}\label{framework}
In this section, we propose a user-centric incentive mechanism framework for edge AIGC services in 6G-IoT networks, as shown in Fig. \ref{System}.

The framework reveals the symbiotic interaction between 6G-IoT networks and edge AIGC. On the one hand, 6G-IoT networks, with their ultra-high-speed connectivity, provide various real-time data collected by IoT devices to the edge devices that deploy trained AIGC models, enabling high-quality model inferences\cite{wen2024generative}. On the other hand, edge AIGC leverages the computation power and intelligence of edge devices to process and analyze collected data and provide real-time decision-making strategies to IoT devices, thus empowering intelligent IoT applications\cite{wen2024generative}. In the following part, we describe the process of incentivizing ASPs for user-centric AIGC service provisions, as illustrated in Part A and Part B of Fig. \ref{System}.
\begin{itemize}
    \item \textbf{\textit{Step 1. Design diffusion-based contracts and request AIGC services:}} Clients first design diffusion-based contracts according to their subjective expectations and the computational resources required for their tasks\cite{10409284}. Then, they upload their generation requests via mobile devices to edge servers, where edge servers that deploy trained AIGC models act as ASPs\cite{liu2023deep}. To meet the demands for AIGC services, the Software-Defined Network (SDN) orchestrator would create virtual connections between ASPs and clients, allowing clients to directly perform transactions with ASPs\cite{liu2023deep}.
    \item \textbf{\textit{Step 2. Select contracts and provide AIGC services:}} Upon receiving service requests from clients, ASPs first select suitable contracts according to the complexity of their AIGC model. Then, the ASPs fine-tune their AIGC models to specific domains using the real-time dataset uploaded by IoT devices and execute inferences based on the client-provided prompts, thus generating desired content and providing AIGC services to clients\cite{wen2023freshness}.
    \item \textbf{\textit{Step 3. Receive AIGC services and give rewards to ASPs:}} Clients receive AIGC services (e.g., text-to-image generation and voice assistants) from ASPs via wireless communication. Subsequently, the clients give rewards specified in the contracts to the ASPs, and the transaction between them takes place through virtual links constructed by the SDN orchestrator\cite{liu2023deep}.
\end{itemize}

Next, we present how diffusion-based user-centric contracts are designed to incentivize ASPs to effectively provide edge AIGC services to clients.

\section{Problem Formulation}\label{Problem}
In this section, we present an incentive mechanism to motivate ASPs to provide edge AIGC services to clients. We first formulate the utility functions of both ASPs and the client. Then, we propose a contract theory model and validate its feasibility.

\subsection{System Model}

We consider that $M$ ASPs with massive connectivity and a client for edge AIGC services. Due to information asymmetry, the client is not aware of the local AIGC model complexity of each ASP precisely\cite{wen2023task}. The client only has the distribution information about the model complexities of ASPs. Thus, according to the complexity of their local AIGC models, the ASPs can be classified into a set $\mathcal{K}= \left\{ \theta_k:1 \leq k \leq K \right\}$ of $K$ types\cite{liu2023deep}. We denote the $k$-th type of ASPs as $\theta_k$. In non-decreasing order, the ASP types are sorted as
$\theta_1 \leq \theta_2 \leq \dots \leq \theta_K$\cite{wen2023task}, where the higher the model complexity, the larger the model size and the better the inference quality\cite{liu2023deep}. To facilitate explanation, the ASP with type $k$ is called the type-$k$ ASP, and the proportion of type-$k$ ASPs in the AIGC market is denoted as $Q_k,\:\forall k\in \mathcal{K}$.
%The number of the type-$n$ worker is $Q_{n}$ and we have $\sum_{ n \in N } Q_{n} = M$.

\subsection{Utilities of AIGC Service Providers and the Client}

According to \cite{liu2023deep}, \textit{the utility of a type-$k$ ASP} is the difference between the evaluation toward the received reward $r(\theta_k, R_k)$ and its resource cost $c(L_k)$ of participating in edge AIGC service provisions, in which $L_k$ and $R_k$ denote the service latency requirement and the reward to the type-$k$ ASP, respectively. Thus, the utility of a type-$k$ ASP is given by\cite{kang2019toward}
\begin{equation}
    \begin{split}
        U_k^A (L_k, R_k) &= r(\theta_k, R_k)-c(L_k)\\
        &= f \theta_k R_k - a\bigg(\frac{L_k}{L_{max}}\bigg),
    \end{split}
\end{equation}
where $f$ is a pre-defined weight parameter about the incentive $R_k$ of type-$k$ ASPs, $a$ is the unit resource cost of edge AIGC service provisions, and $L_{max}$ represents the maximum tolerant service latency\cite{liu2023deep}.

\begin{figure}[t]
\centering
\includegraphics[width=0.4\textwidth]{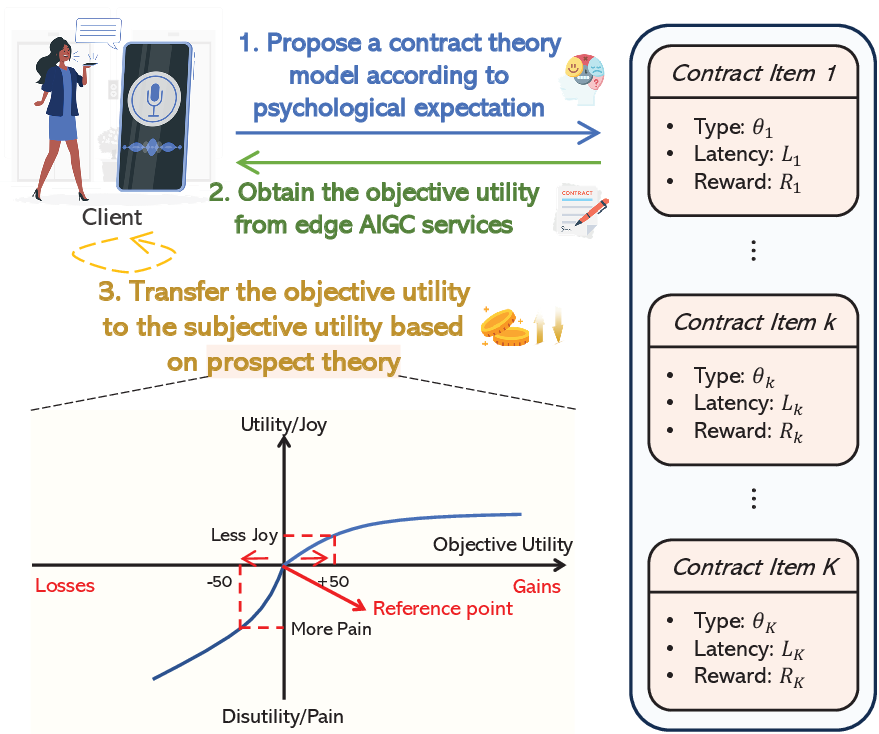}
\captionsetup{font=footnotesize}
\caption{An illustration for the subjective utility of the client based on prospect theory in the contract theory model.}
\label{Prospect}
\end{figure}

Next, we derive \textit{the subjective utility of the client} by using PT. The objective utility of the client towards the type-$k$ ASP is denoted as $U_k^C$, equaling the difference between the revenue for receiving AIGC services within $L_k$ and the reward $R_k$, where the revenue is defined as a general security-latency metric, i.e., $e_1(\theta_k)^{z_1}+e_2(L_k/L_{max})^{z_2}$\cite{liu2023deep,kang2019toward}. Here, $e_1>0$ and $e_2>0$ are pre-defined parameters regarding the quality of AIGC model inference and the latency spent for edge AIGC service delivery\cite{kang2019toward}, respectively. Similarly, $z_1 \geq 1$ and $z_2 \geq 1$ are given factors indicating the effects of inference quality and the latency \cite{kang2019toward}, respectively. This function indicates that the client prefers ASPs with complex AIGC models and high-quality AIGC service delivery within the specified time (i.e., no exceed $L_{max}$). Thus, the objective utility of the client towards type-$k$ ASPs is expressed as
\begin{equation}
    \begin{split}
        U_k^{EUT} = e_1(\theta_k)^{z_1}+e_2\bigg(\frac{L_k}{L_{max}}\bigg)^{z_2}-R_k.
    \end{split}
\end{equation}

Due to information asymmetry, the client only knows the type distribution but cannot exactly know the private type of each ASP, which results in uncertainty when the client makes decisions. Therefore, the client overcomes the information asymmetry problem by adopting the EUT to define its overall objective utility as\cite{kang2023blockchain,liu2023deep}:
\begin{equation}\label{F_EUT}
\begin{split}
U_{EUT}^C  = M{\sum^K_{k=1}} Q_{k} U_{k}^{EUT} ,\\
\end{split}
\end{equation} 
where $\sum_{ k=1 }^K Q_{k} = 1$. 

However, when facing uncertain and risky circumstances, the client may behave irrationally and have different risk attitudes, such as risk seeking or risk averse\cite{kang2023blockchain}. Therefore, EUT is not applicable to capture risk attitudes of the client during the uncertain decision-making process. To address this issue, we utilize PT to further model the objective utility of the client, which makes the contract model more practical, as shown in Fig. \ref{Prospect}.
Given a reference point $U_{ref}$ for all types of ASPs, we convert $U_k^{EUT}$ into the subjective utility, which is given by\cite{huang2021efficient, ye2022incentivizing,kang2023blockchain} 
\begin{equation} \label{PT}
U_{k,PT}^C= \left\{ \begin{aligned}
( U_k^{EUT}- U_{ref}  )^{\zeta^{+}}, \: U_k^{EUT}\geq U_{ref},\\
- \eta (  U_{ref} - U_k^{EUT}   )^{\zeta^{-}},\: U_k^{EUT}< U_{ref},\\
\end{aligned} \right.
\end{equation}
where $\zeta^{+}, \zeta^{-}\in (0,1]$ are the weighting factors representing gain and loss distortion, respectively, and $\eta\geq 0$  is a loss aversion coefficient.
% where the parameters $\delta^{+}, \delta^{-} \in (0,1]$ and $\xi \geq 1$ model the risk aversion and loss aversion, respectively.
Based on (\ref{PT}), the overall subjective utility of the client is given by
\begin{equation}\label{F_PT}
\begin{split}
U_{PT}^C = M{\sum^K_{k=1}} Q_{k} U_{k, PT}^C.\\
\end{split}
\end{equation}

\subsection{Contract Formulation}

In the AIGC service market, the ASP types are private information that may not be visible to the client, indicating that there exists an information asymmetry between the client and the ASPs. Contract theory as an economic tool is powerful for designing incentive mechanisms with asymmetric information\cite{wen2023freshness,wen2023task}, which can effectively address such information asymmetry by designing specific contracts for every type of ASPs\cite{liu2023deep}. Here, the client is the leader in designing a contract with a group of contract items, and each ASP selects the best contract item according to its type. The contract item is denoted as $\Phi = \left\{ (L_{k}, R_{k}), k \in \mathcal{K}\right\}$. To ensure that each ASP automatically selects the contract item designed for its specific type, the feasible contract should satisfy the following Individual Rationality (IR) and Incentive Compatibility (IC) constraints\cite{kang2023blockchain}.

\begin{definition}
(Individual Rationality) The type-$k$ ASP achieves a non-negative utility by selecting the contract item $(L_k, R_k)$ corresponding to its type, i.e.,
\begin{equation}\label{IR2}
    \begin{split}
       U_k^A(L_k, R_k)\geq 0, \:\forall k \in \mathcal{K}.
    \end{split}
\end{equation}
\end{definition}

\begin{definition}
(Incentive Compatibility) An ASP of any type prefers to select the contract item $(L_{k}, R_{k})$ designed for
its type rather than any other contract item $(L_{n}, R_{n}),\: n\in \mathcal{K}$, and $n\neq k$, i.e.,
\begin{equation}
    \begin{split}\label{IC1}
        U_k^A(L_k, R_k) \geq U_k^A(L_n, R_n), \:\forall n, k \in \mathcal{K},\: n\neq k.
    \end{split}
\end{equation}
\end{definition}

The IR constraints can encourage ASPs to participate in AIGC service provisions by ensuring that the utility of ASPs is non-negative, and the IC constraints can guarantee that each ASP provides high-quality AIGC services by choosing the optimal contract item designed for its type. Based on the IR and IC constraints, the problem of maximizing the overall subjective utility of the client is formulated as
\begin{equation}\label{problem1}
    \begin{split}
        & \max_{\bm{L}, \bm{R}} U_{PT}^C \\
        %& \quad\:\:\:\text{s.t.} \: (\ref{IR2})\: \text{and}\: (\ref{IC1}),\\
        &\:\:\text{s.t.}\: U_k^A(L_k, R_k)\geq 0, \:\forall k \in \mathcal{K},\\
        &\quad\:\:\:\: U_k^A(L_k, R_k) \geq U_k^A(L_n, R_n), \:\forall n, k \in \mathcal{K},\: n\neq k,\\
        &\quad\:\:\:\: L_k \geq 0, R_k \geq 0, \theta_k > 0,\: \forall k \in \mathcal{K},
    \end{split}
\end{equation}
where $\bm{L} =[L_{k}]_{1 \times K}$ and $\bm{R} =[R_{k}]_{1 \times K}$. 
% There are $N$ IR constraints and $N(N-1)$ IC constraints, making it quite difficult to  solve \textbf{{Problem 1}} directly (\ref{problem1}). 

\subsection{Contract Feasibility}
In this part, we focus on validating the feasibility of the proposed contract and obtaining the optimal reward $R_k^*$ to type-$k$ ASPs. We first derive the following necessary conditions that can be derived from the IR and IC constraints.
\begin{lemma}\label{lem1}
    For any feasible contract $(\bm{L},\bm{R})$, if $\theta_i > \theta_j$, then $R_i > R_j,\: \forall i, j \in \mathcal{K}$.
\end{lemma}
\begin{proof}
    Utilizing the contradiction method, if $R_i < R_j$ when $\theta_i > \theta_j$, we have 
    \begin{equation}\label{lem1_3}
        \begin{split}
            (R_i - R_j)(\theta_i - \theta_j) < 0.
        \end{split}
    \end{equation}
    Based on the IC constraints (\ref{IC1}), we can obtain
    \begin{equation}\label{lem1_1}
        \begin{split}
            f \theta_i R_i - a\bigg(\frac{L_i}{L_{max}}\bigg)\geq f \theta_i R_j - a\bigg(\frac{L_j}{L_{max}}\bigg),
        \end{split}
    \end{equation}
    \begin{equation}\label{lem1_2}
        \begin{split}
             f \theta_j R_j - a\bigg(\frac{L_j}{L_{max}}\bigg)\geq f \theta_j R_i - a\bigg(\frac{L_i}{L_{max}}\bigg).
        \end{split}
    \end{equation}
    Combining (\ref{lem1_1}) and (\ref{lem1_2}), we have
        \begin{equation}
        \begin{split}
            \theta_i R_i+\theta_j R_j - \theta_i R_j - \theta_j R_i \geq 0,
        \end{split}
    \end{equation}
    \begin{equation}\label{lem1_4}
        \begin{split}
             (\theta_i - \theta_j)(R_i - R_j) \geq 0.
        \end{split}
    \end{equation}
    Since (\ref{lem1_3}) is contradiction with (\ref{lem1_4}), we can obtain that if $\theta_i > \theta_j$, then $R_i > R_j$. Thus, the proof is completed.
\end{proof}
\textbf{Lemma \ref{lem1}} indicates that a higher type of ASPs, that is, a more complex AIGC model, leads to higher rewards.

\begin{lemma}\label{lem2}
    For any feasible contract $(\bm{L},\bm{R})$, $R_i > R_j$ if and only if $L_i > L_j,\: \forall i, j \in \mathcal{K}$.
\end{lemma}
\begin{proof}
    First, we prove that if $L_i > L_j$, then $R_i > R_j$. Based on the IC constraints (\ref{IC1}), we have
    \begin{equation}\label{lem2_1}
        \begin{split}
            f \theta_i R_i - a\bigg(\frac{L_i}{L_{max}}\bigg)\geq f \theta_i R_j - a\bigg(\frac{L_j}{L_{max}}\bigg).
        \end{split}
    \end{equation}
    Then, we obtain
    \begin{equation}\label{lem2_2}
        \begin{split}
            f \theta_i (R_i-R_j) \geq a\bigg(\frac{L_i-L_j}{L_{max}}\bigg).
        \end{split}
    \end{equation}
    Since $L_i > L_j$, we have
        \begin{equation}\label{lem2_3}
        \begin{split}
            f \theta_i (R_i-R_j) \geq a\bigg(\frac{L_i-L_j}{L_{max}}\bigg)> 0.
        \end{split}
    \end{equation}
    Thus, $R_i > R_j$.
    Next, we prove that if $R_i > R_j$, then $L_i > L_j$. Based on the IC constraints (\ref{IC1}), we have
    \begin{equation}\label{lem2_4}
        \begin{split}
            f \theta_j R_j - a\bigg(\frac{L_j}{L_{max}}\bigg)\geq f \theta_j R_i - a\bigg(\frac{L_i}{L_{max}}\bigg).
        \end{split}
    \end{equation}
    Then, we obtain
    \begin{equation}\label{lem2_5}
        \begin{split}
            a\bigg(\frac{L_i-L_j}{L_{max}}\bigg) \geq f \theta_j (R_i-R_j).
        \end{split}
    \end{equation}
    Since $R_i > R_j$, we have
    \begin{equation}\label{lem2_6}
        \begin{split}
            a\bigg(\frac{L_i-L_j}{L_{max}}\bigg) \geq f \theta_j (R_i-R_j) > 0,
        \end{split}
    \end{equation}
    and thus $L_i > L_j$. The proof is completed.
\end{proof}

\textbf{Lemma \ref{lem2}} shows that under the role of contract, the higher service latency requirement to the ASP,
%which indicates the ASP has more complex AIGC models and better quality inference services, 
the higher reward received by the ASP.

\begin{lemma}\label{lem3}
    For any feasible contract $(\bm{L},\bm{R})$, $R_i = R_j$ if and only if $L_i = L_j,\: \forall i, j \in \mathcal{K}$.
\end{lemma}

Since the proof of \textbf{Lemma \ref{lem3}} is similar to that of \textbf{Lemma \ref{lem2}}, we omit it for brevity. \textbf{Lemma \ref{lem3}} indicates that the ASPs with the same service latency requirement will obtain the same amount of rewards, together with \textbf{Lemma \ref{lem1}} and \textbf{Lemma \ref{lem2}}, we can define the following relationship.

\begin{definition}\label{mono}
    (Monotonicity) If $\theta_i \geq \theta_j,\:\forall i, j \in \mathcal{K}$, we have $R_i \geq R_j$.
\end{definition}

We next proceed with the elimination of IC and IR constraints. Before that, the IC constraints between type-$i$ and type-$m$ ASPs, where $m \in \{1,2,\ldots,i-1\}$, are defined as the Downward Incentive Constraints (DICs)\cite{hou2017incentive}, and the IC constraints between type-$i$ and type-$n$ ASPs, where $n \in\{i+1, i+2, \ldots, K\}$, are defined as the Upward Incentive Constraints (UICs)\cite{hou2017incentive}. Especially, the IC constraint between type-$i$ and type-$i-1$ ASPs is defined as the Local Downward Incentive Constraint (LDIC), and the 
IC constraint between type-$i$ and type-$i+1$ ASPs is defined as the Local Upward Incentive Constraint (LUIC). Based on the above analysis, we can obtain the following lemmas.

\begin{lemma}\label{lem4}
    (Individual Rationality Constraint Elimination) Based on the IC constraints (\ref{IC1}), the IR constraints can be reduced as
    \begin{equation}
        \begin{split}
            f \theta_1 R_1 - a\bigg(\frac{L_1}{L_{max}}\bigg)\geq 0.
        \end{split}
    \end{equation}
\end{lemma}
\begin{proof}
    Based on the IC constraints, we can obtain
    \begin{equation}\label{lem4_1}
        \begin{split}
             f \theta_k R_k - a\bigg(\frac{L_k}{L_{max}}\bigg)\geq f \theta_k R_1 - a\bigg(\frac{L_1}{L_{max}}\bigg),
        \end{split}
    \end{equation}
    where $\forall k \in \{2,\ldots, K\}$. Since $\theta_1\leq\theta_2\leq\cdots\leq\theta_K$, we have
        \begin{equation}\label{lem4_2}
        \begin{split}
             f \theta_k R_1 - a\bigg(\frac{L_1}{L_{max}}\bigg)\geq f \theta_1 R_1 - a\bigg(\frac{L_1}{L_{max}}\bigg),
        \end{split}
    \end{equation}
    Combining (\ref{lem4_1}) and (\ref{lem4_2}), we have
    \begin{equation}\label{lem4_3}
        \begin{split}
             f \theta_k R_k - a\bigg(\frac{L_k}{L_{max}}\bigg)\geq f \theta_1 R_1 - a\bigg(\frac{L_1}{L_{max}}\bigg)\geq 0.
        \end{split}
    \end{equation}
    Based on (\ref{lem4_3}), if the IR constraint of the type-$1$ ASPs holds, the IC constraints for the rest $(K-1)$ types of ASPs can be eliminated. Thus, the proof is completed.
\end{proof}

\begin{lemma}\label{lem5}
    (Incentive Compatibility Constraint Elimination) With monotonicity, the IC constraints can be reduced as the LDIC, i.e.,
    \begin{equation}
        \begin{split}
             f \theta_i R_i - a\bigg(\frac{L_i}{L_{max}}\bigg)\geq f \theta_i R_{i-1} - a\bigg(\frac{L_{i-1}}{L_{max}}\bigg),
        \end{split}
    \end{equation}
    where $\forall i \in \{2,\ldots, K\}$, and the LUIC, i.e.,
    \begin{equation}
        \begin{split}
             f \theta_i R_i - a\bigg(\frac{L_i}{L_{max}}\bigg)\geq f \theta_i R_{i+1} - a\bigg(\frac{L_{i+1}}{L_{max}}\bigg).
        \end{split}
    \end{equation}
\end{lemma}
\begin{proof}
    The $K(K-1)$ IC constraints (\ref{IC1}) can be divided into $\frac{K(K-1)}{2}$ DICs, i.e.,
    \begin{equation}\label{lem5_1}
        \begin{split}
             f \theta_i R_i - a\bigg(\frac{L_i}{L_{max}}\bigg)\geq f \theta_i R_{j} - a\bigg(\frac{L_{j}}{L_{max}}\bigg),
        \end{split}
    \end{equation}
    where $\forall i, j \in \{2,\ldots, K\},\: i>j$, and $\frac{K(K-1)}{2}$ UICs, i.e.,
    \begin{equation}\label{lem5_2}
        \begin{split}
             f \theta_i R_i - a\bigg(\frac{L_i}{L_{max}}\bigg)\geq f \theta_i R_{j} - a\bigg(\frac{L_{j}}{L_{max}}\bigg),
        \end{split}
    \end{equation}
    where $\forall i, j \in \{2,\ldots, K\},\: i<j$. For $\theta_{i-1} < \theta_i < \theta_{i+1},\: \forall i \in \{2,\ldots,K-1\}$, we have
    \begin{equation}\label{lem5_4}
        \begin{split}
            f \theta_{i+1} R_{i+1} - a\bigg(\frac{L_{i+1}}{L_{max}}\bigg)\geq f \theta_{i+1} R_{i} - a\bigg(\frac{L_{i}}{L_{max}}\bigg),
        \end{split}
    \end{equation}
        \begin{equation}
        \begin{split}
            f \theta_{i} R_{i} - a\bigg(\frac{L_{i}}{L_{max}}\bigg)\geq f \theta_{i} R_{i-1} - a\bigg(\frac{L_{i-1}}{L_{max}}\bigg).
        \end{split}
    \end{equation}
    Based on \textbf{Definition} \ref{mono}, we can further obtain
    \begin{equation}
        \begin{split}
            f(\theta_{i+1}-\theta_i)(R_i-R_{i-1})\geq0,
        \end{split}
    \end{equation}
    \begin{equation}
        \begin{split}
            f\theta_{i+1}(R_i-R_{i-1})\geq f\theta_i(R_i-R_{i-1}).
        \end{split}
    \end{equation}
    Thus, we have
    \begin{equation}\label{lem5_3}
        f\theta_{i+1}(R_i-R_{i-1})\geq f\theta_{i}(R_i-R_{i-1})\geq a\bigg(\frac{L_{i}-L_{i-1}}{L_{max}}\bigg).
    \end{equation}
    Similarly, (\ref{lem5_3}) is rewritten as
    \begin{equation}\label{lem5_3}
        f \theta_{i+1} R_{i} - a\bigg(\frac{L_{i}}{L_{max}}\bigg)\geq f \theta_{i+1} R_{i-1} - a\bigg(\frac{L_{i-1}}{L_{max}}\bigg).
    \end{equation}
    Combining (\ref{lem5_4}) and (\ref{lem5_3}), we can obtain
    \begin{equation}\label{lem5_5}
        \begin{split}
        f \theta_{i+1} R_{i+1} - a\bigg(\frac{L_{i+1}}{L_{max}}\bigg)\geq f \theta_{i+1} R_{i-1} - a\bigg(\frac{L_{i-1}}{L_{max}}\bigg).
        \end{split}
    \end{equation}
    By using (\ref{lem5_5}), we can extend downward until $\theta_1$, i.e.,
    \begin{equation}\label{lem5_6}
        \begin{split}
            f \theta_{i+1}& R_{i+1} - a\bigg(\frac{L_{i+1}}{L_{max}}\bigg)\geq f \theta_{i+1} R_{i-1} - a\bigg(\frac{L_{i-1}}{L_{max}}\bigg)\\
            &\geq \cdots \geq  f \theta_{1} R_{1} - a\bigg(\frac{L_{1}}{L_{max}}\bigg),\: \forall i \in \{2,\ldots, K\}.
        \end{split}
    \end{equation}
    Based on (\ref{lem5_6}), we can deduce that the validity of DICs is ensured by the presence of monotonicity and the LDIC. Similarly, it can be demonstrated that the existence of monotonicity and LUIC guarantees the fulfillment of UICs. Thus, the proof is completed.
\end{proof}

According to \textbf{Lemmas} \ref{lem4} and \ref{lem5}, we can deduce the satisfaction conditions of the feasible contract as follows:
\begin{lemma}\label{lem6}
    With information asymmetry, the feasible contract should satisfy the following conditions:
     \begin{subequations}
        \begin{align}
            &\:f \theta_{1} R_{1} - a\bigg(\frac{L_{1}}{L_{max}}\bigg) = 0,\\
            &\:f \theta_{i} (R_{i}-R_{i-1}) = a\bigg(\frac{L_{i}-L_{i-1}}{L_{max}}\bigg),\: \forall i \in \{2,\ldots,K\},\\
            &\: R_K \geq R_{K-1} \geq \cdots \geq R_1,\: L_K \geq L_{K-1}\geq \cdots \geq L_1.
        \end{align}
    \end{subequations}
\end{lemma}
\begin{proof}
    For $f \theta_{1} R_{1} - a\big(\frac{L_{1}}{L_{max}}\big) \geq 0$, the client would reduce the reward $R_1$ as much as possible until $f \theta_{1} R_{1} - a\big(\frac{L_{1}}{L_{max}}\big) = 0$ to maximize its utility. For the LDIC, the client would also reduce $R_i$ as much as possible until $f \theta_{i} R_{i} - a\big(\frac{L_{i}}{L_{max}}\big) = f \theta_{i} R_{i-1} - a\big(\frac{L_{i-1}}{L_{max}}\big)$, thus maximizing its utility. In the following part, we prove that if $f \theta_{i} R_{i} - a\big(\frac{L_{i}}{L_{max}}\big) = f \theta_{i} R_{i-1} - a\big(\frac{L_{i-1}}{L_{max}}\big),\: \forall i \in \{2,\ldots, K\}$ and the monotonicity hold, the LUIC holds.

    If $\theta_i \geq \theta_{i-1}$, then $R_i \geq R_{i-1}$, and we further have
    \begin{equation}
        f \theta_{i} (R_i - R_{i-1}) \geq f \theta_{i-1}(R_i - R_{i-1}). 
    \end{equation}
    Since $f \theta_i (R_i-R_{i-1}) = a\big(\frac{L_i-L_{i-1}}{L_{max}}\big)$, we have
    \begin{equation}
        f\theta_{i-1}(R_i-R_{i-1})\leq a \bigg(\frac{L_i}{L_{max}}-\frac{L_{i-1}}{L_{max}}\bigg),
    \end{equation}
    and similarly, we have
    \begin{equation}
        f\theta_{i-1}R_{i-1} - a\bigg(\frac{L_{i-1}}{L_{max}}\bigg) \geq f \theta_{i-1} R_i - a\bigg(\frac{L_i}{L_{max}}\bigg).
    \end{equation}
    Thus, the LDIC remains unchanged when both $R_i$ and $R_{i-1}$ decrease by the same value. The proof is completed.
\end{proof}

Based on \textbf{Lemma \ref{lem6}}, we can derive the optimal reward to ASPs for edge AIGC service provisions by using the iterative method, given by
\begin{equation}\label{R_n2}
    \begin{aligned}
        R_k^{*} = \frac{a}{fL_{max}}\frac{L_1}{\theta_1} + \frac{a}{fL_{max}} \sum\limits_{i=1}^k\Delta_i,\:k\in \mathcal{K}, 
    \end{aligned}
\end{equation}
where $\Delta_1 = 0$ and $\Delta_i = \frac{L_i-L_{i-1}}{\theta_i},\: i \in \{2,\ldots,k\}$.

\begin{figure*}[ht]
\centering
\captionsetup{font=footnotesize}
\includegraphics[width=0.95\textwidth]{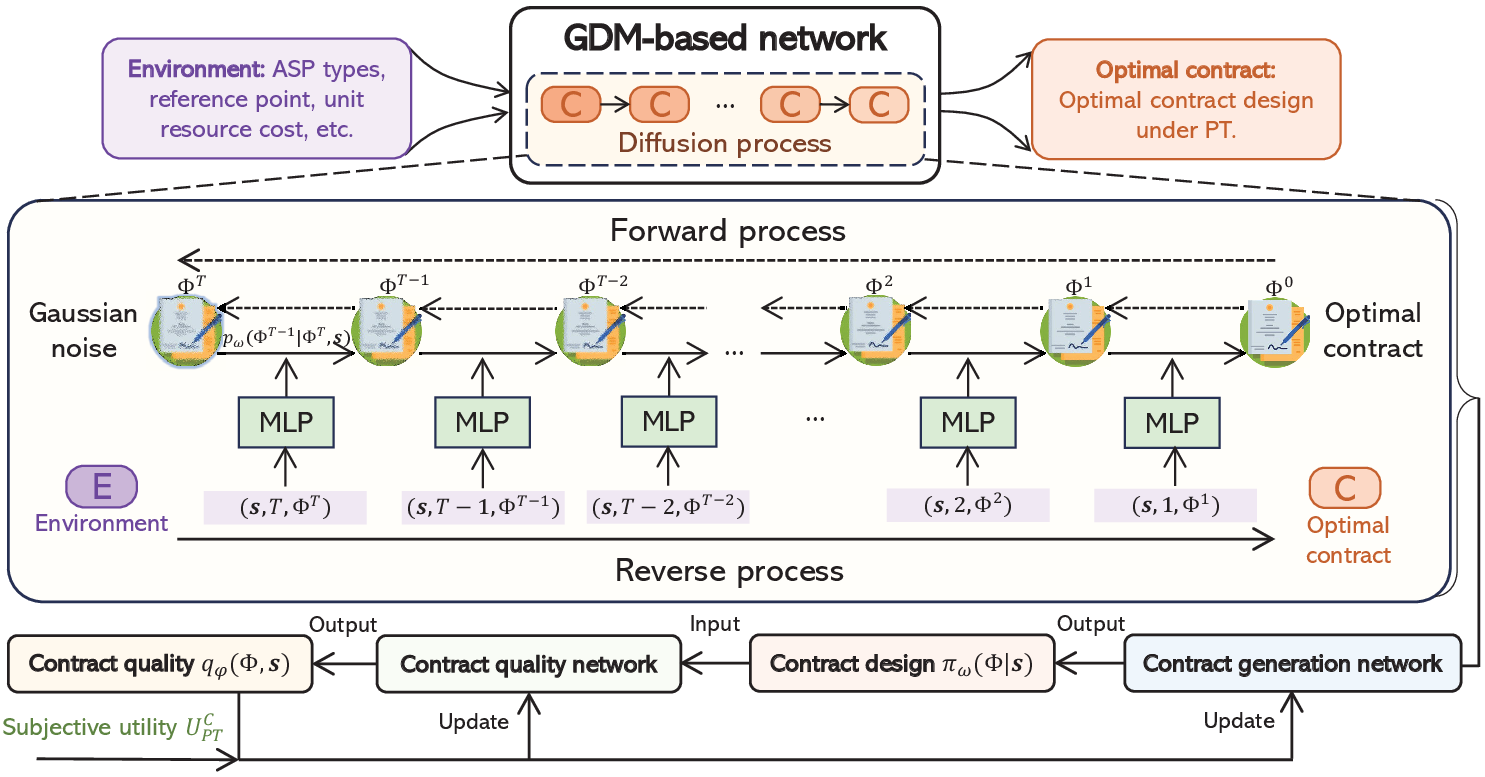}
% \caption{The entire process of UT migration in UAV metaverses.}
\caption{Generative diffusion models for optimal contract design under prospect theory. Note that MLP refers to the multi-layer perception\cite{10409284}.}
\label{diffusion}
\end{figure*}

\begin{remark}
We demonstrate that an optimal contract in the proposed model exists and validate its feasibility. The optimal contract design under PT is modeled as a complex decision-making problem. While traditional mathematical techniques may face difficulties in solving this particular problem, DRL-based approaches may offer a feasible alternative\cite{10409284}. Particularly, Proximal Policy Optimization (PPO) is an on-policy and model-free algorithm that iteratively updates the policy towards optimal performance by using a clipped surrogate objective\cite{wen2024learning}, and SAC stands out as an off-policy algorithm that trains a stochastic policy to maximize both the expected cumulative reward and policy entropy\cite{haarnoja2018soft}, achieving exceptional performance in a wide range of decision-making tasks. However, in practice, the dynamics of 6G IoT networks may significantly impact the action and state spaces of DRL models\cite{lai2023resource,wen2024generative}, necessitating the complete retraining of DRL models and affecting their efficiency of finding the optimal contract design\cite{wen2024generative}. Fortunately, GDMs can effectively capture the high-dimensional and complex structures of 6G IoT networks and present superior performance in decision-making problems\cite{du2023deep}, which are adopted for the optimal contract design under PT in Section \ref{Optimial_Contract}.
\end{remark}

\section{Generative Diffusion Models for Optimal Contract Design under Prospect Theory}\label{Optimial_Contract}
In this section, we propose the diffusion-based SAC algorithm to generate the optimal contract design under PT, as shown in Fig. \ref{diffusion}. Then, we analyze the computational complexity of the proposed algorithm.

% Introduce PT theory

%%%%%%%%%%%% algorithm
% \subsection{Generative Diffusion Models for Contract Generation}
In comparison to traditional backpropagation algorithms in neural networks or DRL techniques, which concentrate on optimizing model parameters directly, GDMs aim to improve contract design through an iterative process of denoising the initial distribution\cite{du2023deep}. Motivated by \cite{10409284}, we advocate for using the diffusion-based SAC algorithm to identify an optimal contract item, i.e., $(L_{k, PT}^{*}, R_k^*),\: k \in \mathcal{K}$, and demonstrate the application of GDMs in generating optimal contract designs\cite{10409284}. Over $T$ iterations, Gaussian noise is systematically added to the initial contract sample $\Phi^0$, resulting in a series of samples $(\Phi^1,\Phi^2,\ldots,\Phi^T)$. In contract modeling, the environment that affects the optimal contract design is defined as 
\begin{equation}
    \boldsymbol{s}:=[M, K, L_{max}, U_{ref}, a, (Q_1,\ldots,Q_K), (\theta_1,\ldots,\theta_K)]. 
\end{equation}

The diffusion model network that maps environmental states to contract designs constitutes the \textit{contract design policy}, denoted as $\pi_\omega(\Phi|\boldsymbol{s})$ with parameters $\omega$. The objective of $\pi_\omega(\Phi|\boldsymbol{s})$ is to generate an optimal contract design that maximizes the expected cumulative reward over a series of time steps. The contract design policy $\pi_\omega(\Phi|\boldsymbol{s})$ through the reverse process of the conditional diffusion model is expressed as \cite{10409284}
\begin{equation}
\begin{split}
    \pi_\omega(\Phi|\boldsymbol{s})&=p_\omega(\Phi^0, \ldots,\Phi^T|\boldsymbol{s})\\
    &=\mathcal{N}(\Phi^T;\textbf{0},\textbf{I})\prod_{t=1}^Tp_\omega(\Phi^{t-1}|\Phi^t,\boldsymbol{s}),
\end{split}
\end{equation}
and the outcome of the reverse chain process represents the selected contract design. $p_\omega(\Phi^{t-1}|\Phi^t,\boldsymbol{s})$ can be modeled as a Gaussian distribution $\mathcal{N}(\Phi^{t-1};\boldsymbol{\mu}_\omega(\Phi^t,\boldsymbol{s},t),\boldsymbol\Sigma_\omega(\Phi^t,\boldsymbol{s},t))$\cite{ho2020denoising}, where the covariance matrix $\boldsymbol\Sigma_\omega(\Phi^t,\boldsymbol{s},t)$ is given by
\begin{equation}
    \boldsymbol\Sigma_\omega(\Phi^t,\boldsymbol{s},t)=\delta_t\textbf{I},
\end{equation}
and the mean $\boldsymbol\mu_\omega(\Phi^t,\boldsymbol{s},t)$ is expressed as
\begin{equation}
\begin{split}
    \boldsymbol\mu_\omega(\Phi^t,\boldsymbol{s},t)=\frac{1}{\sqrt{\chi_t}}\bigg(\Phi^t-\frac{\delta_t}{\sqrt{1-\bar{\chi_t}}}\boldsymbol\epsilon_\omega(\Phi^t,\boldsymbol{s},t)\bigg).
\end{split}
\end{equation}
Here, $\boldsymbol\epsilon_\omega$ represents the \textit{contract generation network}. Based on (\ref{sample}), we first sample $\Phi^T \sim \mathcal N (\textbf{0}, \textbf{I})$ and then sample from the reverse diffusion chain parameterized by $\omega$, given by\cite{10409284}
\begin{equation}\label{denoise}
    \Phi^{t-1}|\Phi^t=\frac{\Phi^t}{\sqrt{\chi_t}}-\frac{\delta_t}{\sqrt{\chi_t(1-\bar{\chi_t})}}\boldsymbol\epsilon_\omega(\Phi^t,\boldsymbol{s},t)+\sqrt{\delta_t}\boldsymbol\epsilon.
\end{equation}

\begin{algorithm}[ht]
\DontPrintSemicolon
\SetAlgoLined
    \caption{Diffusion-based SAC Algorithm for Optimal Contract Design under PT}\label{AI_Contract}

    \KwIn{Diffusion step $T$, batch size $N$, discount factor $\gamma$, soft target update parameter $\tau$.}
    \KwOut{The optimal contract design $\Phi^0$.}
    \textbf{\textit{Phase 1: Initialization} }\\
    Initialize replay buffer $\mathcal{D}$, contract generation network $\boldsymbol{\epsilon}_\omega$ with weights $\omega$, contract quality network $q_\varphi$ with weights $\varphi$, target contract generation network $\boldsymbol{\epsilon}'_{\omega'}$ with weights $\omega'$, target contract quality network $q'_{\varphi'}$ with weights $\varphi'$.\\
    
    \textbf{\textit{Phase 2: Training}} \\
    \For{\rm{Episode} $=1$ to Max\_episode $M_e$}
    {
        Initialize a random process $\mathcal{N}$ for contract design exploration. \\
        \For{\rm{Step} $z=1$ to Max\_step $M_z$}
        {
            Observe the current environment $\boldsymbol{s}_z$.\\
            Set $\Phi_z^T$ as Gaussian noise and generate contract design $\Phi_z^0$ by denoising $\Phi_z^T$ using $\boldsymbol{\epsilon}_\omega$ based on  (\ref{denoise}). \\
            Execute contract design $\Phi_z^0$ and observe the reward $r_z$. \\
            Store record $(\boldsymbol{s}_z,\Phi_z^0, r_z, \boldsymbol{s}_{z+1})$ into replay buffer $\mathcal{D}$. \\
            Sample a random mini-batch $\mathcal{B}_z$ of $N$ records $(\boldsymbol{s}_i,\Phi_i,r_i)$ from replay buffer $\mathcal{D}$. \\
            Update the contract quality network by minimizing (\ref{Q_update}). \\
            Update the contract generation network by computing the policy gradient (\ref{policy_update})\\
            %$\nabla_\omega\boldsymbol{\epsilon}_\omega\approx\frac{1}{N}\sum_i\nabla_{\Phi^0}q_\varphi(\boldsymbol{s},\Phi^0)|_{\boldsymbol{s}=\boldsymbol{s}_i}\nabla_\omega\boldsymbol{\epsilon}_\omega|\boldsymbol{s}_i.$\\
            Update the target networks:
            $\omega'\leftarrow\tau\omega+(1-\tau)\omega'$,
            $\varphi'\leftarrow\tau \varphi+(1-\tau)\varphi'$. \\
        }
    }
    \textbf{return} The trained contract generation network $\boldsymbol{\epsilon}_\omega$. \\
    \textbf{\textit{Phase 3: Inference}} \\
    Input the environment vector $\boldsymbol{s}:=[M, K, L_{max}, U_{ref}, a, (Q_1,\ldots,Q_K), (\theta_1,\ldots,\theta_K)]$. \\
    Generate the optimal contract design $\Phi^0$ by denoising Gaussian noise using $\boldsymbol{\epsilon}'_{\omega'}$ based on (\ref{denoise}). \\
    \textbf{return} $\Phi^0=\small\{(L_{k,PT}^{*},R_k^*),\:k \in \mathcal{K} \small\}$.
\end{algorithm}

\begin{table}[ht]\label{parameter}
	\renewcommand{\arraystretch}{1.2}
        \captionsetup{font = small}
	\caption{ Summary of Key Training Hyperparameters. }\label{table} \centering %\tabcolsep=5pt
	\begin{tabular}{m{5.7cm}<{\raggedright}|m{1.7cm}<{\centering}}	 	
		\hline		
		\textbf{Hyperparameters} & \textbf{Setting}\\	
		\hline
		Learning rate of the contract generation network &  $2\times10^{-7}$\\	
		\hline
		Learning rate of the contract quality network  &  $2\times10^{-6}$  \\	
		\hline
		Soft target update parameter ($\tau$)   & $0.005$  \\
		\hline
		Batch size ($N$) &  $512$\\	
		\hline		
		Discount factor ($\gamma$)&  $0.99$\\
		\hline
        Denoising steps for the diffusion model ($T$) &  $5$\\
		\hline
            Maximum capacity of the replay buffer ($\mathcal{D}$) &  $ 10^6$\\
		\hline
	\end{tabular}\label{parameter}
\end{table}
%\begin{equation}
%    \pi=\arg\min_{\pi_\omega}\mathcal{L}(\omega)=-\Bbb E_{\Phi^0\sim \pi_\omega}[q_\varphi(\boldsymbol{s},\Phi^0)].
%\end{equation}
The effective training of the contract design policy $\pi_\omega$ in the high-dimensional and complex environment $\boldsymbol{s}$ can facilitate the training of the contract generation network $\boldsymbol\epsilon_\omega$. Thus, motivated by the Q-function in DRL, we define the contract quality network as $q_\varphi (\boldsymbol{s}, \Phi)$, which associates an environment-contract pair $\{\boldsymbol{s}, \Phi\}$, using a value to represent the expected cumulative reward for adhering to the contract design policy from the current state\cite{10409284}. Therefore, the optimal contract design policy is the policy that maximizes the expected cumulative utility of the client, given by \cite{10409284}
\begin{equation}\label{policy_update}
    \pi = \arg\max_{\pi_{{\omega}}} \Bbb E \Bigg[\sum_{z = 0} ^ Z \gamma^z (r(\bm{s}_z,\Phi_z)-\varsigma\pi_{\omega}(\bm{s}_z)\log\pi_{\omega}(\bm{s}_z))\Bigg],
\end{equation}
where $\gamma$ represents the discount factor for future rewards and $\varsigma$ denotes the temperature coefficient controlling the strength of the entropy. $r(\bm{s}_z,\Phi_z)$ is the immediate reward upon executing action $\Phi_z$ in state $\bm{s}_z$. If the action $\Phi_z$ both satisfies the IC and IR constraints, $r(\bm{s}_z,\Phi_z)$ is represented as
\begin{equation}
\begin{split}
    r(\bm{s}_z,\Phi_z) = &U_{PT}^C + \sum_{k=1}^KU_k^A(L_k, R_k)\\
&+\sum_{k=1}^K\sum_{n=1}^K\Big(U_k^A(L_k, R_k)-U_k^A(L_n, R_n)\Big).
\end{split}
\end{equation}
Otherwise, $r(\bm{s}_z,\Phi_z)$ is set to $0$.
\begin{comment}
Here, the loss function $\mathcal{L}(\omega)$ is defined as
\begin{equation}\label{loss}
\begin{split}
    \mathcal{L}(\omega)=\frac{1}{N}\sum_{i=1}^N\big(U_{PT}^C+\gamma q_{\varphi '}'(\boldsymbol{s}_i,\Phi_i^{'0})-q_\varphi(\boldsymbol{s}_i,\Phi_i)\big)^2,\\ 
\end{split}
\end{equation}
where $N$ in the context of random mini-batch sampling refers to the number of records sampled from the replay buffer $\mathcal{D}$ for training neural networks, $\gamma$ represents the discount factor, $q_{\varphi '}'$ represents the target contract quality network of $q_\varphi$, and $\Phi_i^{'0}$ can be acquired from the target contract generation network $\boldsymbol\epsilon '_{\omega '}$.
\end{comment}

The contract quality network is trained through traditional methods, employing the double Q-learning technique to minimize the Bellman operator \cite{hasselt2010double}. This process involves two networks, denoted as $q_{\varphi_1}$ and $q_{\varphi_2}$, alongside corresponding target networks, namely $q_{\varphi '_1}$, $q_{\varphi '_2}$, and $\pi_{\omega '}$. Subsequently, the optimization of $\varphi_i$ for $i = 1,2$ is conducted by minimizing the objective function as follows\cite{haarnoja2018soft}:
%\begin{equation}
%    \Bbb E_{\Phi^0\sim\pi_{\omega '}}\big [\big\|\big(r(\boldsymbol{s},\Phi_t)+\gamma\min_{i=1,2}q_{\varphi '_i}(\boldsymbol{s},\Phi^0_{t+1})\big)-q_{\varphi_i}(\boldsymbol{s},\Phi_t)\big\|^2\big ].
%\end{equation}
\begin{equation}\label{Q_update}
\begin{split}
    &\Bbb E_{(\bm{s}_z, \Phi_z, \bm{s}_{z+1}, r_z)\sim \mathcal{B}_z}\Big[\sum_{m=1,2}(r(\bm{s}_z,\Phi_z)-q_{\varphi_m}(\bm{s}_z,\Phi_z)\\
    &\quad\quad+\gamma^z(1-d_{z+1})\pi_{\omega '}(\bm{s}_{z+1})q'_{\varphi'}(\bm{s}_{z+1}))^2\Big],
\end{split}
\end{equation}
where $\mathcal{B}_z$ is a mini-batch of transitions sampled from the experience replay memory $\mathcal{D}$ for training neural networks in the training step $z$, $d_{z+1}$ is a $0$-$1$ variable representing the terminated flag, and $q'_{\varphi'}(\bm{s}_{z+1})) = \min\{q_{\varphi '_1}(\bm{s}_{z+1}), q_{\varphi '_2}(\bm{s}_{z+1})\}$ represents the target contract quality network\cite{10409284}.
% Specifically, the GDM-based generative contract algorithm entails employing denoising techniques to formulate a contract design, followed by the addition of exploration noise to the contract design and its execution to accumulate exploration experience. The algorithm for generating contracts based on GDM is depicted in \textbf{Algorithm \ref{AI_Contract}}. Utilizing \textbf{Algorithm \ref{AI_Contract}} enables us to derive the optimal contract that maximizes the utility of the client under PT.

%This process succeeded by integrating exploration noise into the contract design and executing it to amass exploration experience.

The proposed diffusion-based SAC algorithm involves utilizing denoising techniques to generate the optimal contract design under PT\cite{10409284}, as illustrated in \textbf{Algorithm \ref{AI_Contract}}, which facilitates the derivation of an optimal contract that maximizes the overall subjective utility of the client under PT. \textbf{Algorithm \ref{AI_Contract}} consists of three phases, and its computational complexity is $\mathcal{O}(|\omega|+|\varphi|+M_e M_z (T|\omega| + |\varphi|))$. Specifically, in the initialization phase, the computational complexity stands at $\mathcal{O}(|\omega|+|\varphi|)$. In the training phase, the computational complexity is $\mathcal{O}(M_e M_z (T|\omega| + |\varphi|))$. Finally, to generate an optimal contract item through the contract generation policy, the computational complexity of the inference phase is $\mathcal{O}(|\omega|)$.

% Motivated by the above analysis, the detailed optimal contract design is shown in \textbf{Algorithm \ref{AI_Contract}}. Firstly, the EUT-based solution $f_{n,EUT}^*$ can be obtained by (\ref{f_k,EUT_min}). Then, based on the above three cases, we compare the sizes of $U_{s,n,EUT}$ and $U_{ref}$ sequentially to obtain the optimal PT-based solution $f_{n,PT}^*$. Finally, by substituting $f_{n,PT}^*$ into (\ref{pi}), the optimal rewards $R_n^*$ can be calculated. In particular, the computational complexity of \textbf{Algorithm \ref{AI_Contract}} in the worst case is $\mathcal{O}(N(N-1))$, which emphasizes that we can use \textbf{Algorithm \ref{AI_Contract}} to find the optimal contract under PT for all cases that are analyzed above.
%the service provider solves the PT-based solution $f_{n,PT}^{*}$ of each contract item based on preference profiles. By substituting $f_{n,PT}^{*}$ into (\ref{pi}), the monetary rewards in the contract items can be calculated. 
%If a worker is the first one to accept a contract item by responding to the service provider, the task will be assigned to the worker, and the other contract items will be automatically invalid. 

\begin{figure*}[t]
    \begin{center}
	\begin{minipage}[t]{0.45\linewidth}
		\centering
            \captionsetup{font=footnotesize}
            \includegraphics[width=1\textwidth]{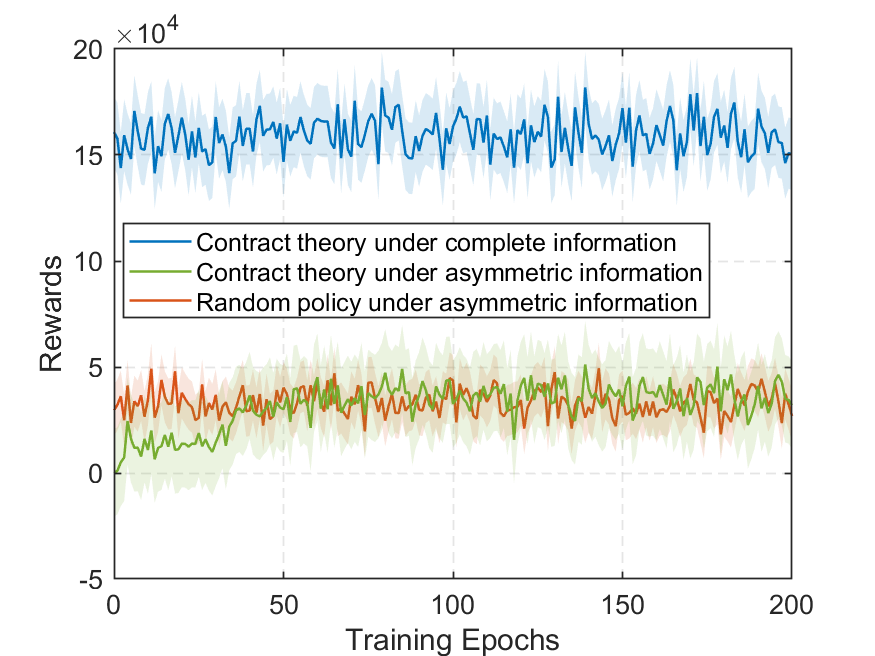}
            \caption{Test reward comparison of our scheme with other schemes under PT, i.e., contract-based incentive mechanism with complete information and random policy under asymmetric information, where reference point $U_{ref}=200$ and loss aversion $\eta = 0.5$.}\label{compare}
	\end{minipage}
	\hspace{0.3in}
	\begin{minipage}[t]{0.45\linewidth}
		\centering
            \includegraphics[width=1\linewidth]{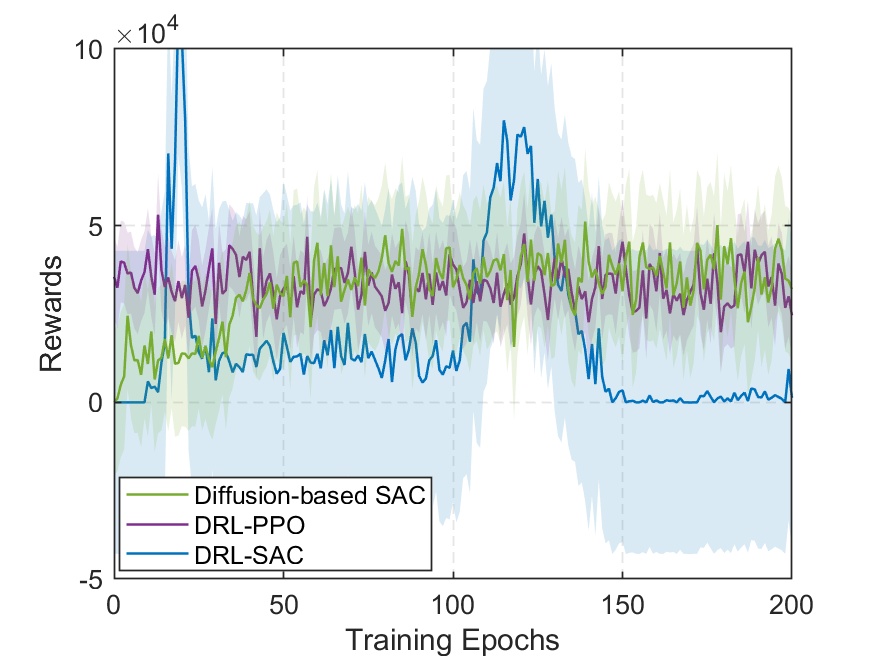}
		\captionsetup{font=footnotesize}
            \caption{Performance comparison between the diffusion-based SAC algorithm and traditional DRL algorithms in optimal contract design under PT, where reference point $U_{ref}=200$ and loss aversion $\eta = 0.5$.}\label{algorithm_compare}
	\end{minipage}
	\end{center}
\end{figure*}

\begin{figure*}[t]
    \begin{center}
	\begin{minipage}[t]{0.45\linewidth}
		\centering
		\captionsetup{font=footnotesize}
		\includegraphics[width=1\linewidth]{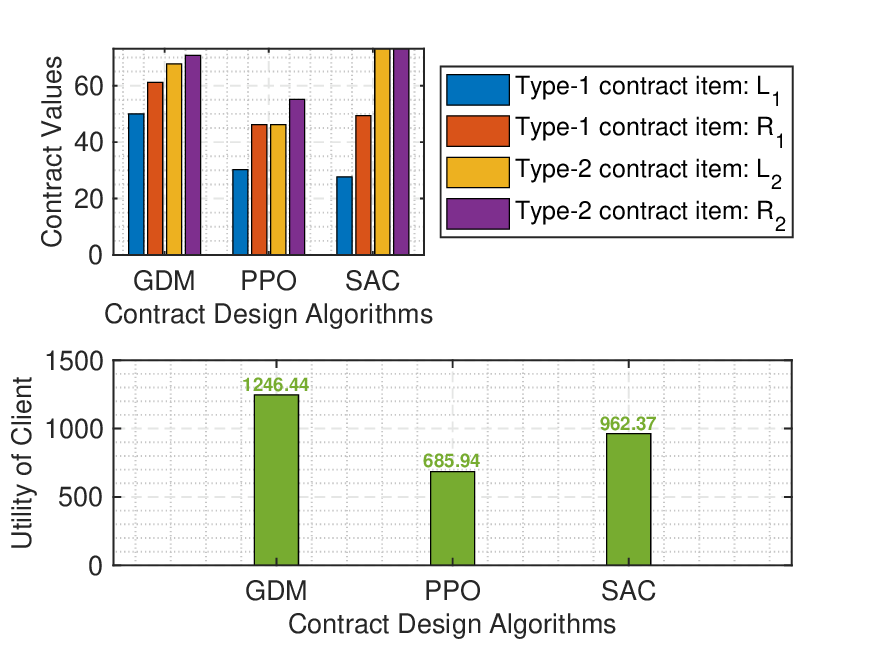}
		\caption{Optimal contracts designed by the proposed diffusion-based SAC, DRL-PPO, and DRL-SAC, with reference point $U_{ref} = 250$ and loss aversion $\eta = 0.5$.}\label{Contract}  
	\end{minipage}
	\hspace{0.3in}
	\begin{minipage}[t]{0.45\linewidth}
		\centering
            \includegraphics[width=1\linewidth]{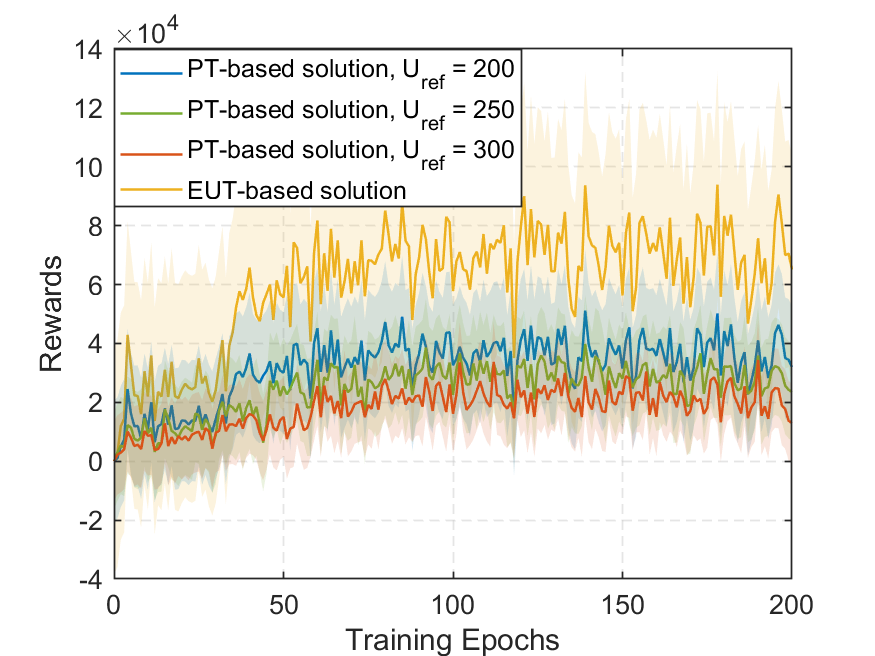}
		\captionsetup{font=footnotesize}
            \caption{Test reward comparison of the proposed scheme under different reference points $U_{ref}$, with loss aversion $\eta = 0.5$.}\label{reference}
	\end{minipage}
	\end{center}
\end{figure*}

\section{Numerical Results}\label{Results}

In this section, we first introduce the setting of the GDM in optimal contract design under PT. Then, we evaluate the performance of the proposed GDM-based contract generation model under PT. 

\subsection{Setting of the GDM}
We design the contract generation network $\boldsymbol{\epsilon}_\omega$ and the contract quality network $q_\varphi$ with the same structure to reduce the issue of overestimation\cite{10409284}. For the GDM-based contract generation model under PT, the specific settings of training hyperparameters are summarized in Table \ref{parameter}. In this paper, we consider $5$ ASPs, which are divided into $2$ types, i.e., $M = 5$ and $K = 2$. According to \cite{wen2024generative, liu2023deep}, $\theta_1$ and $\theta_2$ are randomly sampled within $[10,50]$ and $[100,200]$, respectively. $Q_1$ and $Q_2$ are generated randomly, following the Dirichlet distribution. The maximum tolerant service latency $L_{max}$ is randomly sampled within $[100,160]$. The unit resource cost of edge AIGC service provisions $a$ is randomly sampled within $[80,100]$. The weighting factors representing gain and loss distortion $\zeta^{+}, \zeta^{-}$ are both set to $1$. Weighting factors $f$, $e_1$, $e_2$, $z_1$, and $z_2$ are set to $0.05$, $3$, $50$, $1$, and $1$, respectively. Our experiments are conducted using PyTorch on NVIDIA GeForce RTX 3080 Laptop GPU with CUDA 12.0.

\begin{figure*}[t]
\centering
\subfigure[Comparison of test reward curves, with reference point $U_{ref} = 200$.]
{
    \begin{minipage}[t]{0.45\linewidth}
	\centering
	\includegraphics[width=1\linewidth]{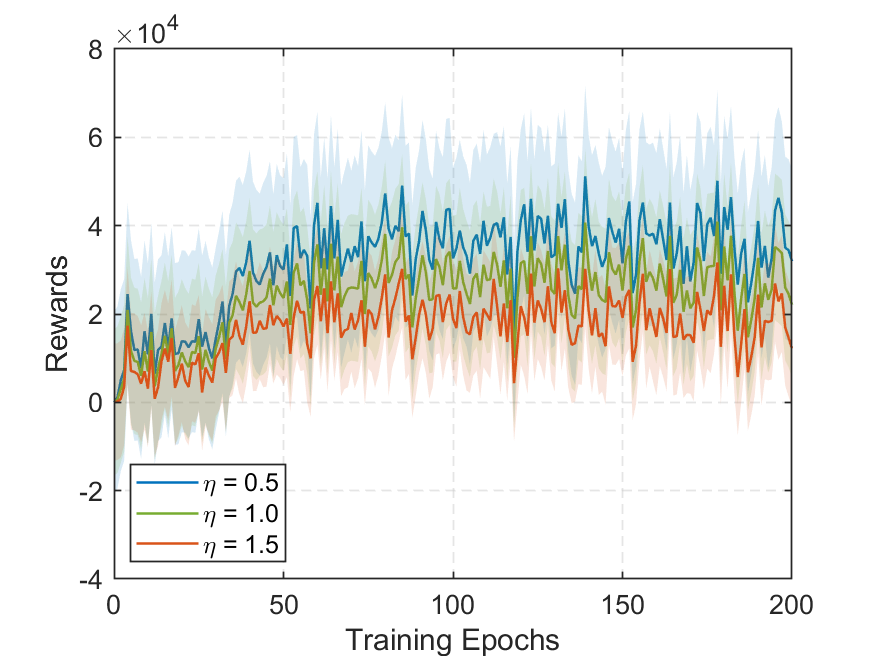}
        \captionsetup{font=footnotesize}
	\label{loss_eta}
    \end{minipage}
}\hspace{0.3in}
\subfigure[Utilites of the client and ASPs.]
{
    \begin{minipage}[t]{0.45\linewidth}
	\centering
	\includegraphics[width=1\linewidth]{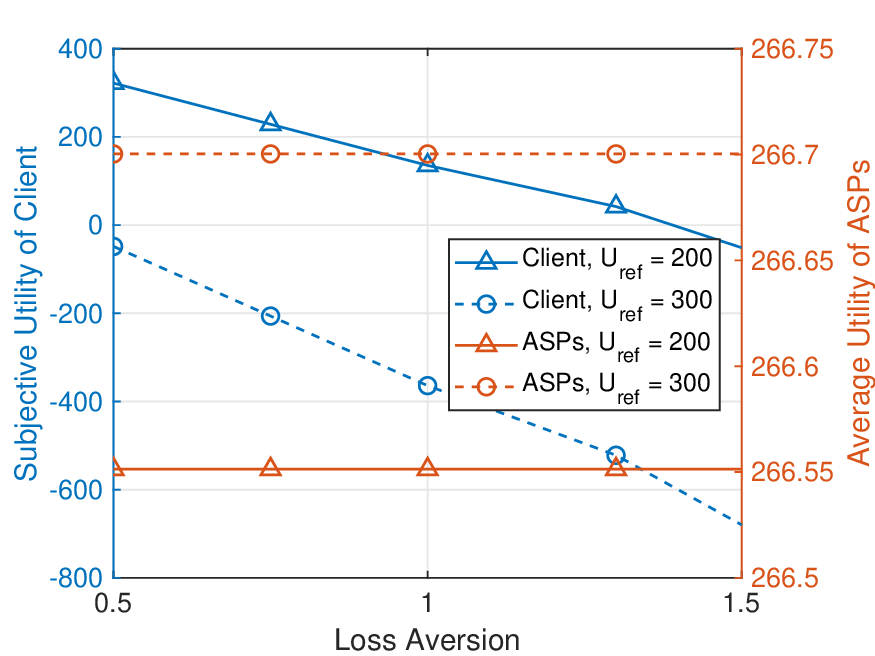}
        \captionsetup{font=footnotesize}
	\label{utility}
    \end{minipage}
}
\captionsetup{font=footnotesize}
\caption{Effect comparison of the proposed algorithm under different loss aversions $\eta$.}\label{effect}
\end{figure*}

\subsection{Performance Analysis of the Proposed GDM-based Contract Generation Model under PT}
To evaluate the performance of the proposed incentive mechanism, we compare the proposed contract-based incentive mechanism with asymmetric information with other schemes under PT: 1) \textit{Contract-based incentive mechanism with complete information}\cite{hou2017incentive} that the client knows the private information of ASPs (i.e., ASP types); 2) \textit{Random policy under asymmetric information} that the client randomly designs contracts regardless of the types of ASPs. As shown in Fig. \ref{compare}, the proposed contract generation model under PT always outperforms random schemes. For the same parameter settings, the performance of the contract-based incentive mechanism with complete information is always greater than that of the proposed contract generation model, which indicates that the client obtains fewer benefits due to information asymmetry. The reason is that the client knows the exact type information of ASPs with complete information, thus offering the most suitable contract item for each ASP\cite{kang2023blockchain}. However, the environment of complete information is not realistic. Although the proposed scheme can effectively mitigate the impact of information asymmetry by leveraging contract theory\cite{wen2023task}, a rational ASP still can maliciously provide false information to manipulate and obtain higher rewards, ultimately reducing the subjective utility of the client. In summary, our proposed contract model under PT allows the client to achieve the highest utility under asymmetric information, which is more reliable in practice.

Figure \ref{algorithm_compare} illustrates the performance comparison between the diffusion-based SAC algorithm and traditional DRL algorithms in optimal contract design under PT. We can observe that by struggling to satisfy IC and IR contracts, the test reward value of the proposed algorithm rises sharply and then converges. Under the same parameter settings, the final test reward obtained by the diffusion-based SAC algorithm is significantly higher than those obtained through DRL algorithms, enabling clients to consistently obtain more utilities, where the final test rewards obtained by the diffusion-based SAC algorithm, DRL-PPO, and DRL-SAC are $31960.22$, $24608.47$, and $1257.813$, respectively. The superior performance of the proposed algorithm is mainly attributed to two aspects\cite{liu2023deep,du2023deep}. Firstly, the fine-grained policy tuning in the diffusion process can effectively alleviate the influence of randomness and noise. Secondly, exploration through diffusion can improve the flexibility and robustness of the strategy and prevent the model from falling into suboptimal solutions. Therefore, this superior performance demonstrates the ability of GDMs to capture intricate patterns and connections among environmental observations\cite{du2023deep}, thus promisingly paving the way for enhanced decision-making in 6G-IoT networks, despite the complicated and high-dimensional nature of the 6G-IoT network environment.

Figure \ref{Contract} presents the quality of contracts designed by the GDM, DRL-PPO, and DRL-SAC, with reference point $U_{ref} = 250$ and loss aversion $\eta = 0.5$. For a given state of the environment, we can observe that contributed to the exploration experience during the denoising process\cite{liu2023deep}, the proposed GDM-based contract generation model under PT can provide a contract design that achieves a client utility value of $1246.44$, which is significantly higher than the $685.94$ achieved by DRL-PPO and the $193.47$ achieved by DRL-SAC. The reason is that the ability of GDMs to generate a near-optimal contract design that allows the client to achieve a higher utility. For the contract design, we can observe that as the type of ASPs increases, the rewards received by the ASPs also increase, and correspondingly, the service latency requirement becomes higher. This serves to validate the contentions made in \textbf{Lemma \ref{lem1}} and \textbf{Lemma \ref{lem2}}. Moreover, we can see that the variables in the contract item are the same when ASPs are the type-$2$ in DRL-SAC, indicating that DRL-SAC can easily fall into the local optima when solving for the optimal contract\cite{du2023deep}. Overall, the above numerical analysis demonstrates the feasibility and superior performance of the proposed scheme.

Figure \ref{reference} shows the impacts of reference points on PT and EUT-based solutions for the optimal contract design, where the PT-based solution refers to the proposed contract model with PT and the EUT-based solution refers to the proposed contract model without PT. By comparing the performance of PT and EUT-based solutions for the proposed contract model, we can observe that the EUT-based solution always outperforms the PT-based solution for the optimal contract design, regardless of reference points $U_{ref}$. The reason is that the EUT-based solution does not consider the irrational behavior of clients when the clients face a risky and uncertain environment\cite{ye2022incentivizing}, diminishing the impact on the utility of clients due to their irrationality. However, this solution may not be feasible in practice. In addition, we can observe that the performance of the PT-based solution with a low reference point is better than that of the PT-based solution with a high reference point, indicating that the clients with low reference points can obtain more revenues theoretically.

Figure \ref{effect} illustrates the effect comparison of the proposed algorithm under different loss aversions $\eta$. As shown in Fig. \ref{loss_eta}, we can observe that the smaller the loss aversion, the better the performance of the proposed algorithm. Moreover, despite differences in loss aversions, GDMs can still effectively find the optimal contracts with the same convergence rate, which demonstrates the adaptability and robustness of the proposed scheme. The adaptability and robustness of the proposed scheme are particularly crucial in 6G-IoT networks where conditions may be unpredictable and complicated\cite{du2023deep}. Figure \ref{utility} shows the impacts of preference parameters on the utilities of the client and ASPs, namely the reference point $U_{ref}$ and loss aversion $\eta$. We can observe that given a reference point, with the loss aversion increasing, the subjective utility of the client decreases. The reason is that the increase in loss aversion means that the client is more inclined to have a risk-averse behavior\cite{kang2023blockchain}, indicating that the client needs ASPs with more complex AIGC models for edge AIGC service provisions to avoid utility losses. Thus, the client needs to give higher rewards to ASPs with higher types, reducing the subjective utility of the client. From Fig. \ref{utility}, we can see that the average utility of ASPs is stable. Although the objective utility of ASPs with higher types increases, the objective utility of ASPs with lower types decreases due to the indifference of the client, resulting in the stability of the average utility of ASPs. When the loss aversion is fixed, with the increase of the reference point, the subjective utility of the client decreases, which validates the results presented in Fig. \ref{loss_eta}.

\section{Conclusion}\label{Conclusion}
In this paper, we studied incentive mechanism design for edge AIGC service provisions within 6G-IoT networks. We designed a user-centric incentive mechanism framework for edge AIGC services in 6G-IoT networks. Specifically, we proposed a contract theory model to incentivize ASPs to provide AIGC services to clients under information asymmetry. Considering the decision-making of clients under risks and uncertainty, we utilized PT to better capture the subjective utility of clients. Furthermore, we proposed the diffusion-based SAC algorithm for optimal contract design under PT. Finally, numerical results demonstrate the effectiveness and reliability of the proposed GDM-based contract generation model under PT. For future work, we will further focus on the incentive mechanism design among multiple clients and multiple ASPs for edge AIGC service provisions, considering the effects of the irrational behavior of clients.

\bibliographystyle{IEEEtran}
\bibliography{ref}

\end{document}